%% file: main.tex
\documentclass[final]{siamart250211}
\usepackage{amsmath}
\usepackage{amsfonts}
\usepackage[english]{babel}
\usepackage{graphicx}

\usepackage{svg}

\usepackage[noend]{algpseudocode}
\usepackage{verbatim}
\usepackage{cleveref}

\usepackage{color, soul}
\usepackage[]{todonotes}
\usepackage{enumitem}

\newtheorem{observation}{Observation}

\newcommand{\notshow}[1]{}

\newcommand\floor[1]{\lfloor#1\rfloor}
\newcommand\ceil[1]{\lceil#1\rceil}
\newtheorem{strategy}{Strategy}

\title{Simple Compact Monotone Tree Drawings
\thanks{A preliminary version of this paper which included the one-quadrant algorithm for non-convex monotone tree drawings was presented in~\cite{DBLP:conf/gd/OikonomouS17}.}
}

\author{
Argyris Oikonomou\thanks{Department of Computer Science, Yale University, U.S.A. (\email{argyris.oikonomou@yale.edu}) }
\and 
Antonios~Symvonis\thanks{School of Applied Mathematical \& Physical Sciences, National Technical University of Athens, Greece. (\email{symvonis@math.ntua.gr}) }
}

\headers{Simple Compact Monotone Tree Drawings}{Anargyros. Oikonomou, and Antonios Symvonis}
\date{}

\begin{document}
\maketitle

\begin{abstract}
A monotone drawing of a graph G is a straight-line  drawing of G such that every pair of vertices is connected by a path that is monotone with respect to some direction. 

Trees, as a special class of graphs, have been the focus of several research papers. He and He~\cite{mt:4} were the first to achieve quadratic drawing area by 
showing how to produce a   monotone drawing of an arbitrary $n$-vertex tree that is contained in a $12n \times 12n$ grid.

All monotone tree drawing algorithms that have appeared in the literature consider rooted ordered trees and they draw them so that (i) the root of the tree is drawn at the origin of the drawing, (ii) the drawing is confined in the first quadrant, and (iii) the ordering/embedding of the tree is respected. In this paper, we provide a simple algorithm that has the exact same characteristics and, given an  $n$-vertex rooted tree $T$, it outputs a monotone drawing of $T$ that fits in a $n \times n$ grid.  
Furthermore, we extend our approach to generate drawings that are both monotone and convex, within a $n\times n$ grid.

For unrooted ordered trees, we present an algorithm that produces monotone drawings that respect the ordering and  fit in an $(n+1) \times (\frac{n}{2} +1)$ grid, while, for unrooted non-ordered trees we produce monotone drawings of good aspect ratio which fit on a grid of size at most $\floor{\frac{3}{4} \left(n+2\right)} \times \floor{\frac{3}{4} \left(n+2\right)}$.
\end{abstract}

\let\labelitemi\labelitemii


\section{Introduction}
\input{introduction.tex}

\section{Definitions and Preliminaries}
\input{preliminaries.tex}

\section{One-Quadrant ``Traditional'' Monotone  Drawing of Rooted Ordered Trees}
\input{algo_1quad.tex}

\section{Convexification of One-Quadrant Monotone Trees}
\input{algo_1quad_convex.tex}

\section{Two-Quadrants Monotone Unrooted Ordered  Tree Drawing}
\input{algo_2quad.tex}

\section{Four-Quadrants Unrooted Monotone Tree Drawing}
\input{algo_4quad.tex}	

\section{Conclusion and Open Problems}
\input{conclusion.tex}

\newpage
\bibliography{bib_1_sym}

\bibliographystyle{siamplain}
\end{document}

%% file: introduction.tex
\label{sect:introduction}

A \emph{straight-line drawing} $\Gamma$ of a graph G is a mapping of each vertex to a distinct point on the plane and of each edge to a straight-line segment between the vertices.
A path $P=\{p_0,p_1,\ldots,p_n\}$ is \emph{monotone} if there exists a line $l$ such that the projections of the vertices of $P$ on $l$ appear on  $l$ in the same order as on $P$.
A straight-line drawing $\Gamma$ of a graph G is \emph{monotone}, if a \emph{monotone} path connects every pair of vertices.

\emph{Monotone graph drawing} has gained the recent attention of researchers and   several interesting results have appeared.
Given a planar fixed embedding of a planar graph $G$,  a planar monotone drawing of $G$ can be constructed,
but at the cost of some bends on some edges (thus no longer a straight-line drawing)~\cite{m:2}.
In the variable embedding setting, there exists a planar monotone drawing of any planar graph without any bends~\cite{m:4}.

One way to find a monotone drawing of a graph is to simply find a monotone drawing of one of its spanning trees. For that reason, the problem of finding  monotone drawings of trees has been the subject of several recent papers, starting from the work by  {Angelini et al.}~\cite{m:1} which introduced  monotone graph drawings. 
{Angelini et al.}~\cite{m:1} provided two algorithms that used ideas from number theory and more specifically Stern-Brocot trees~\cite{Stern1858,Brocot1860},~\cite[Sect. 4.5]{Graham:1994}.
The first algorithm used a grid of size $O(n^{1.6}) \times O(n^{1.6})$ (BFS-based algorithm) while the second one used a grid of size $O(n) \times O(n^2)$ (DFS-based algorithm).
Later, {Kindermann et al.}~\cite{mt:1} provided an algorithm based on Farey sequence (see~\cite[Sect. 4.5]{Graham:1994}) that used a grid of size $O(n^{1.5}) \times O(n^{1.5})$.
He and He~\cite{mt:2} gave an algorithm based on Farey sequence and  reduced the required grid size   to $O(n^{1.205}) \times O(n^{1.205})$, which was the first result that used  less than $O(n^3)$ area.
Recently, He and He~\cite{mt:3}  firstly reduced the grid size for a monotone tree drawing to $O(n  \log(n)) \times O(n  \log(n))$ and,  in a sequel paper, to   $O(n) \times O(n)$~\cite{mt:4}. Their monotone tree drawing uses a grid of size at most  $12n \times 12n$ which turns out to be asymptotically optimal as there exist trees which require at least ${n\over 9} \times {n \over 9}$ area~\cite{mt:4}.\\

Convex tree drawings \cite{CE:2006} is a special subclass of monotone tree drawings \cite{ArkinCM89}. A tree drawing is \emph{convex} if after extending each edge incident to a leaf to an infinite ray (originating at the leaf's parent and passing through the leaf), we obtain a drawing such that: (i) no two of the rays intersect, and (ii) the rays together with the original tree edges partition the plane into convex regions.

It is well known that a strictly convex tree drawing, obtained by strengthening condition~(ii) so that every edge of each convex face meets its neighbors at strictly convex angles, is monotone (Property~7 in \cite{m:1}). Even without this stronger requirement, convex tree drawings remain monotone under mild assumptions (Lemma~3 in \cite{m:1}). For counterexamples of monotone but non-convex drawings see \cref{fig:non_convex_by_1quad,fig:bad_case_one,fig:bad_case_two} in \Cref{sec:convex_drawings}. Although the drawings we produce in \Cref{sec:convex_drawings} are only weakly convex, so existing monotonicity theorems do not apply directly, we nonetheless prove in \cref{thm:convex_drawing} that they, too, are monotone.

\textbf{Our Contribution:}

All monotone tree drawing algorithms that have appeared in the literature consider rooted ordered trees and they draw them so that (i) the root of the tree is drawn at the origin of the drawing, (ii) the drawing is confined in the first quadrant, and (iii) the embedding of the tree is respected. In this paper, we provide a simple algorithm that has the exact same characteristics and, given an  $n$-vertex rooted tree $T$, it outputs a monotone drawing of $T$ that fits on a $n \times n$ grid. Despite its simplicity, our algorithm improves the $12n \times 12n$  result of  He and He~\cite{mt:4}. 
{In addition, we show how our method can be adapted to produce convex monotone tree drawings that are embedded on an $n \times n$ grid.
}

By relaxing the drawing restrictions we can achieve smaller drawing area. More specifically, by carefully selecting a new root for the tree, which we draw it at the origin, we can produce a ``two-quadrants'' drawing that fits in an $(n+1) \times (\frac{n}{2} +1)$ grid. We note that the produced drawing respects the given embedding of the tree. By further relaxing this requirement, i.e., by allowing to change the order of the neighbors  of a tree vertex around it, we can achieve a drawing of  better aspect ratio and smaller area (compared to our $n \times n$ algorithm). More specifically, we describe a ``four-quandrants'' algorithm that draws an $n$-vertex tree on a grid of size at most $\floor{\frac{3}{4} \left(n+2\right)} \times \floor{\frac{3}{4} \left(n+2\right)}$.

The paper is organized as follows:
\Cref{sect:preliminaries} provides definitions and preliminary results.
In \Cref{sect:1quad_drawing} and \Cref{sec:convex_drawings} we present algorithms for one-quadrant non-convex and convex drawings, respectively. In \Cref{sect:2quad_drawing} and~\Cref{sect:4quad_drawing} we present our algorithms for  two- and four-quadrants algorithms, respectively. 
We conclude in \Cref{sect:conclusions}. 

A preliminary version of this paper which included the one-quadrant algorithm for non-convex monotone tree drawings was presented in~\cite{DBLP:conf/gd/OikonomouS17}.   

%% file: preliminaries.tex
\label{sect:preliminaries}\label{sec:preliminaries}

Let $\Gamma$ be a drawing of a graph $G$ and $(u,v)$ be an edge from vertex $u$ to vertex $v$ in $G$.
The slope of edge $(u,v)$, denoted by $slope(u,v)$, is the angle spanned by a counter-clockwise rotation that brings a horizontal half-line starting at $u$ and directed towards increasing $x$-coordinates to coincide with the half-line starting at $u$ and passing through $v$.
We consider slopes that are equivalent modulo $2\pi$ as the same slope.
Observe that $slope(u,v)=slope(v,u) - \pi$.


Let $T$ be a tree rooted at a vertex $r$.
Denote by $T_v$ the subtree of $T$ rooted at a vertex $v$. By $|T_v|$ we denote the number of vertices of $T_v$.
Let $v$ be a child of $u$. By 
{\color{blue}  $T^u_v$  }
we denote the tree that consists of edge $(u,v)$ and $T_v$. In the rest of the paper, we assume that all tree edges are directed away from the root. A rooted tree is said to be \emph{ordered} if there is an order imposed on the children of each vertex. A drawing is said to  \emph{ respect the ordering} of the tree (or the embedding) if the children of a vertex are drawn around it in the specified order.

When producing  a grid drawing, it is common to refer to the \emph{side-length} of the required grid   and to its  \emph{dimensions}. 
We  emphasize that we measure length (width/height)  in units of distance, but when we denote the dimensions of a grid we use the number of grid points in each dimension. So, a grid of width $w$ 
and height $h$ fits in a $(w+1) \times (h+1)$ grid.

\subsection{Slope-disjoint Tree Drawings}

Angelini et al.~\cite{m:1} defined the notion of slope-disjoint tree drawings.  Let $\Gamma$ be a drawing of a rooted tree $T$. $\Gamma$ is called a \emph{slope-disjoint} drawing of  $T$ if the following properties are satisfied:
\begin{enumerate}
\item For every vertex $u \in T$, there exist two angles $a_1(u)$ and $a_2(u)$, with $0<a_1(u)<a_2(u)<\pi$ such that for every edge $e$ that is either in $T_u$ or that enters $u$ from its parent, it holds that $a_1(u)<slope(e)<a_2(u)$.

\item For every two vertices $u, v \in T$ such that  $v$ is a child of $u$, it holds that $a_1(u) < a_1(v) < a_2(v) < a_2(u)$.
\item For every two vertices $v_1, v_2$ having  the same parent, it holds that either $a_1(v_1) < a_2(v_1) < a_1(v_2) < a_2(v_2)$ or $a_1(v_2) < a_2(v_2) < a_1(v_1) < a_2(v_1)$.
\end{enumerate}

The idea behind the definition of  slope-disjoint tree drawings is that all edges in the subtree $T_u$  as well as the edge entering $u$ from its parent have slopes that \emph{strictly} fall within  the angle-range $\left< a_1(u), a_2(u)\right>$ defined for vertex $u$.
$\left< a_1(u), a_2(u)\right>$ is called the \emph{angle-range of}  $u$ with   $a_1(u)$ and $a_2(u)$ being its \emph{boundaries}.
The convex angle formed between two half-lines with slopes $a_1(u)$ and $a_2(u)$ is denoted by $\phi_u = a_2(u)-a_1(u)$ and is called \emph{angle-range length} of $u$.

Angelini et al.~\cite{m:1} proved the following  theorems:

\begin{theorem}[Angelini et al.\cite{m:1}]\label{t:b}
Every monotone drawing of a tree is planar.
\end{theorem}

\begin{theorem}[Angelini et al.\cite{m:1}]\label{t:a}
Every slope-disjoint drawing of a tree is monotone.
\end{theorem}

In order to simplify the description of our algorithm, we extend the definition of slope-disjoint tree drawings to allow for angle-ranges of adjacent vertices (parent-child relationship) or sibling vertices (children of the same parent)  to share angle-range boundaries. 

\begin{definition}
\label{def:nonStrictlySlopeDisjoint}
A tree drawing $\Gamma$  of a rooted tree $T$ is called a \emph{non-strictly slope-disjoint} drawing if the following properties are satisfied:

\begin{enumerate}

\item
For every vertex $u \in T$, there exist two angles $a_1(u)$ and $a_2(u)$, with 
$0 \boldsymbol{\leq} a_1(u)<a_2(u)\boldsymbol{\leq} \pi$ 
such that for every edge $e$ that is either in $T_u$ or enters $u$ from its parent, it holds that $a_1(u)<slope(e)<a_2(u)$.

\item 
For every two vertices $u, v \in T$ such that $v$ is a child of $u$, it holds that $a_1(u) \boldsymbol{\leq} a_1(v) < a_2(v) \boldsymbol{\leq} a_2(u)$.
  
\item 
For every two vertices $v_1, v_2$ with the same parent, it holds that either $a_1(v_1) < a_2(v_1) \boldsymbol{\leq} a_1(v_2) < a_2(v_2)$ or $a_1(v_2) < a_2(v_2) \boldsymbol{\leq} a_1(v_1) < a_2(v_1)$.
\end{enumerate}
\end{definition}

In our extended definition, we allow for angle-ranges of adjacent vertices (parent-child relationship) or sibling vertices (children of the same parent)  to share angle-range boundaries. Note that replacing the ``$\leq$'' symbols in our definition by the ``$<$'' symbol gives us the original definition of Angelini et al.~\cite{m:1} for the slope disjoint tree drawings.

\begin{lemma} \label{l:sd}
Every non-strictly slope-disjoint drawing of a tree $T$ is also  a slope-disjoint drawing.
\end{lemma}

\begin{proof}
Intuitively, the theorem holds since we can always adjust (by a tiny amount) the angle-ranges of vertices that share an angle-range boundary   so that, after the adjustment  no two   vertices share an   angle-range boundary. Note that the actual drawing of the tree does not change. Only the angle-ranges are adjusted.
 
More formally, let $\Gamma$ be a \emph{non-strictly slope-disjoint} drawing of a tree $T$ rooted at $r$.
We show how to compute for every vertex $u$ a new angle-range $\left< b_1(u),b_2(u)\right>$ such that the current drawing of $T$ with the new angle-range is \emph{slope-disjoint}.

Let $e(u)$ be the edge that connects the parent of $u$ to $u$ in $T$, for $u \in T \backslash r$.

We make use of the following definitions:
\begin{align*}
\delta_1 & =  min_{u \in T \backslash r }\left(slope(e(u)) - a_1(u)\right)\\
\delta_2 & =  min_{u \in T \backslash r }\left(a_2(u) - slope(e(u)\right))\\
\delta   & =  min(\delta_1, \delta_2)
\end{align*}

For any vertex $u \in T \backslash r$ 
and for any edge $e$ that belongs in $T_u$ or enters $u$ from its parent,
it holds that:
\begin{gather}
	slope(e) - a_1(u) \geq \delta \label{eq:d1} \\
	a_2(u) - slope(e) \geq \delta \label{eq:d2}
\end{gather}

By Property-1 of the non-strictly slope-disjoint drawing, we have that $\delta_1, \delta_2 > 0$ and, therefore, $\delta > 0$.
By adding the two previous inequalities we get that,
\begin{equation}
\cref{eq:d1} + \cref{eq:d2} \Rightarrow  a_2(u)-a_1(u) \geq 2\delta  \text{~where~} u \in T \backslash r \label{eq:i1}
\end{equation}

For any descendant $v$ of the root $r$ of  $T$,  by
inductive use of Property-2 of the non-strictly slope-disjoint drawings, it holds that:
\begin{align}
	a_1(r) & \leq a_1(v) \label{eq:a1} \\
	a_2(r) & \geq a_2(v) \label{eq:a2}
\end{align}

By subtracting $\cref{eq:a1}$ from $\cref{eq:a2}$ we get
\begin{equation}
\cref{eq:a2} - \cref{eq:a1} \Rightarrow  a_2(r)-a_1(r) \geq a_2(v) - a_1(v) \overset{\cref{eq:i1}}{\geq} 2\delta  \label{eq:root}\\ 
\end{equation}
Therefore, by \cref{eq:i1}  and \cref{eq:root},  for any vertex $u \in T$ it holds: 
\begin{equation}\label{eq:2d}
a_2(u) - a_1(u) \geq 2\delta 
\end{equation}

Let the root $r$ of $T$ be at level-0, let $u$ be a vertex in level-$i, ~i>0$, and let $h$ be the  height of  tree $T$.  Define the  slope-disjoint angle-ranges $\left< b_1(u), b_2(u)\right>$ for each vertex $u \in T$ as follows:

\begin{minipage}{.45\textwidth}
\begin{equation*}
b_1(u) =
\begin{cases}
a_1(r) & \text{if } u=r
\\
a_1(u) + \delta \cdot \frac{i}{h+1} & \text{if } u \neq r
\end{cases}
\end{equation*}
\end{minipage}
\hfill
\begin{minipage}{.45\textwidth}
\begin{equation*}
b_2(u) =
\begin{cases}
a_2(r) & \text{if } u=r
\\
a_2(u) - \delta \cdot \frac{i}{h+1} & \text{if } u \neq r
\end{cases}
\end{equation*}
\end{minipage}
\linebreak

Firstly, we show that the new angle-range boundaries $b_1(\cdot), ~b_2(\cdot)$ satisfy Property-2 of slope-disjoint drawings.  
Let $u$ be a level-$i$ vertex $\in T$ and  $v$ be its child. By the non-strictly slope-disjoint Property-2, it holds that:
\begin{align*}
a_1(u)\leq a_1(v) & \Rightarrow  a_1(u) + \delta \cdot \frac{i}{h+1} < a_1(v) + \delta \cdot \frac{i+1}{h+1}\\
& \Leftrightarrow b_1(u)<b_1(v)
\end{align*}

Similarly, we have that $b_2(v)<b_2(u)$.
We also have,
\begin{align*}
b_2(v)-b_1(v) & =  a_2(v) - \delta \cdot \frac{i+1}{h+1} - \left( a_1(v) + \delta \cdot \frac{i+1}{h+1} \right)\\
&  =   (a_2(v) - a_1(v)) -2\delta \cdot \frac{i+1}{h+1}\\
& \overset{\cref{eq:2d}} {\geq}   2\delta - 2\delta \cdot \frac{i+1}{h+1}\\ 
&  = 2\delta \cdot \left( 1-\frac{i+1}{h+1}\right)\\
&  =  2\delta \cdot \frac{h-i}{h+1}\\
& >    0\\
& \Rightarrow b_1(v)< b_2(v)
\end{align*}

The last inequality holds since vertex $u$ has a child and, thus, $u$  is at a level $i$ such that  $i<h$.
Thus, Property-2 of slope-disjoint drawings holds.

Secondly, we show that the new angle-range boundaries $b_1(\cdot), ~b_2(\cdot)$ satisfy Property-3 of  slope-disjoint drawings. Let $v_1,$ and $v_2$ be two level-$i$ vertices having the same parent. Then, by Property-3 of the non-strictly slope disjoint drawings we have that $a_1(v_1) < a_2(v_1) \leq a_1(v_2) < a_2(v_2)$ or $a_1(v_2) < a_2(v_2) \leq a_1(v_1) < a_2(v_1)$.
The two cases are symmetric, so we only prove that $b_1(v_1) < b_2(v_1) < b_1(v_2) < b_2(v_2)$ when $a_1(v_1) < a_2(v_1) \leq a_1(v_2) < a_2(v_2)$.
As proved for the case of Property-2, $b_1(v_1) < b_2(v_1)$ and $b_1(v_2) < b_2(v_2)$ and thus, it remains to prove that $b_2(v_1) < b_1(v_2)$. But we have that,
\begin{align*}
a_2(v_1)\leq  a_1(v_2) &  \Rightarrow   a_2(v_1) - \delta \cdot \frac{i}{h+1} < a_1(v_2) + \delta \cdot \frac{i}{h+1}\\
& \Leftrightarrow   b_2(v_1) < b_1(v_2)
\end{align*}

Finally, we turn our attention to Property-1 of slope-disjoint drawings.
Angle-range boundaries $a_1(\cdot)$ and $ a_2(\cdot)$ satisfy Property-1 of non-strictly slope-disjoint drawings and thus, for every  vertex $u$ at level $i$ and  for every edge $e$ that belongs in $T_u$ or that enters $u$ from its parent inequality $\cref{eq:d1}$ holds.
By definition, we have that $b_1(u)  = a_1(u) +  \delta \cdot \frac{i}{h+1}$ which implies
\begin{equation} \label{eq:2}
a_1(u) = b_1(u) - \delta \cdot \frac{i}{h+1}
\end{equation}
\begin{align*}
\cref{eq:d1} & \overset{\cref{eq:2}}{\Leftrightarrow}  slope(e) - \left(b_1(u) - \delta \cdot \frac{i}{h+1}\right)  \geq \delta\\ 
& \Leftrightarrow  slope(e) - b_1(u) \geq \delta \cdot \left(1-\frac{i}{h+1} \right)\\
& \Leftrightarrow  slope(e) - b_1(u) \geq \delta \cdot \left(\frac{h+1-i}{h+1}\right)\\
& \Rightarrow  slope(e) - b_1(u) > 0\\
\end{align*}

The last inequality holds since  $\delta >0$ and  $ i < h+1$. The later is true since  $u$ is a level-$i$ vertex where  
$i \leq h$.

In a similar way, we show that  $b_2(u) - slope(e)>0$ and  we conclude that $b_1(u) < slope(e) < b_2(u)$. 
Thus,  Property-1 of slope-disjoint drawing is also satisfied.
\end{proof}

\begin{theorem}\label{thm:NSslopeDis_monotonePlanar}
Every non-strictly slope-disjoint drawing of a tree is monotone and planar.
\end{theorem}

\begin{proof}
By \Cref{l:sd} every non-strictly slope-disjoint drawing of a tree $T$ is slope-disjoint and by \Cref{t:b} and \Cref{t:a} it is monotone and planar.
\end{proof}

\subsection{Locating Points on the Grid}

Based on geometry, we now prove that it is always possible to identify points on a grid that satisfy several properties with respect to their location.

\begin{lemma}\label{lemma:smallAngle}
Consider two angles $\theta_1$, $\theta_2$ with $0 \leq \theta_1 < \theta_2 \leq \frac{\pi}{4}$ and let $d=\ceil{\frac{1}{\theta_2-\theta_1}}$. 
Then, edge $e$ connecting the origin $(0,0)$ to point $p=(d, \floor{tan(\theta_1) \cdot d +1})$ satisfies 
$\theta_1 < slope(e) < \theta_2$.
\end{lemma}

\begin{proof}
Let $l_1$ and $l_2$ be the half-lines from origin with slopes $\theta_1$ and $\theta_2$, respectively. 
Let $p_1$ and $p_2$ be the intersection points of $l_1$ and $l_2$ with line $x=d$, respectively. 
We  prove that $|p_1 p_2|>1$, so a grid point must lie between $p_1$ and $p_2$, since the $x$-coordinate is integer and line segment $p_1 p_2$ is parallel to $y$-axis as seen in \Cref{fig:d2}.

From trigonometry, we know  identities :
\begin{equation}
	tan(a - b)  = \frac{tan(a)-tan(b)}{1+tan(a) \cdot tan(b)} \label{eq:t1}
\end{equation}
and
\begin{equation}
	tan(a)  > a \text{ , when } 0 < a < \frac{\pi}{2} \label{eq:t2}
\end{equation}

By $\cref{eq:t1}$, it holds that  $tan(a)-tan(b)=tan(a-b) \cdot (1+tan(a)tan(b))$ and thus,  
for $0 \leq a,b \leq \frac{\pi}{2}$ it holds: 
\begin{equation}
tan(a)-tan(b) > tan(a-b) \text{, when } 0 \leq a,b \leq \frac{\pi}{2}  \label{eq:i2}
\end{equation}

The coordinates of point $p_1$ are $(d,tan(\theta_1) \cdot d)$ while  the coordinates of point $p_2$ are $(d,tan(\theta_2) \cdot d)$. Therefore, 
\begin{align*}
|p_1 p_2|& =  tan(\theta_2) \cdot d - tan(\theta_1) \cdot d\\
	& =  (tan(\theta_2)-tan(\theta_1)) \cdot d\\
	& \overset{\cref{eq:i2}} {>}  tan(\theta_2-\theta_1) \cdot d\\
	& \overset{\cref{eq:t2}} {\geq}   (\theta_2-\theta_1) \cdot d\\
	& = (\theta_2-\theta_1) \cdot \ceil{\frac{1}{\theta_2-\theta_1}}\\
	& \geq  (\theta_2-\theta_1) \cdot \frac{1}{\theta_2-\theta_1}\\
	& =  1
\end{align*} 

Given that $|p_1 p_2|>1$, the grid point $p=(d, \floor{tan(\theta_1) \cdot d +1})$ falls within the angular sector defined by half-lines $l_1$ and $l_2$ and satisfies the lemma. 
\end{proof}

\begin{figure}[tb]
	\centering
	\begin{minipage}{0.45\textwidth}
		\centering
		\includegraphics[scale=0.3]{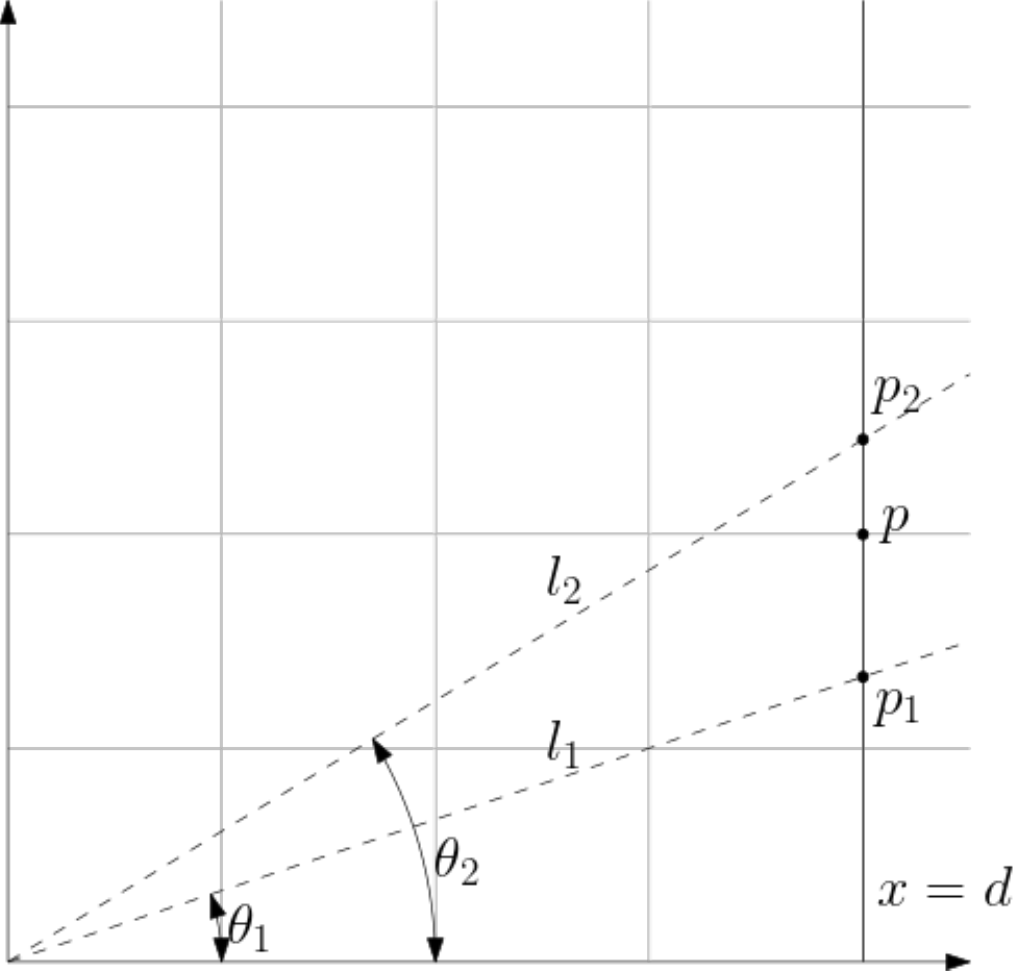}
		\caption{Geometric representation of \Cref{lemma:smallAngle}.}
		\label{fig:d2}
	\end{minipage}
\hfill
	\begin{minipage}{0.45\textwidth}
		\centering
		\includegraphics[scale=0.3]{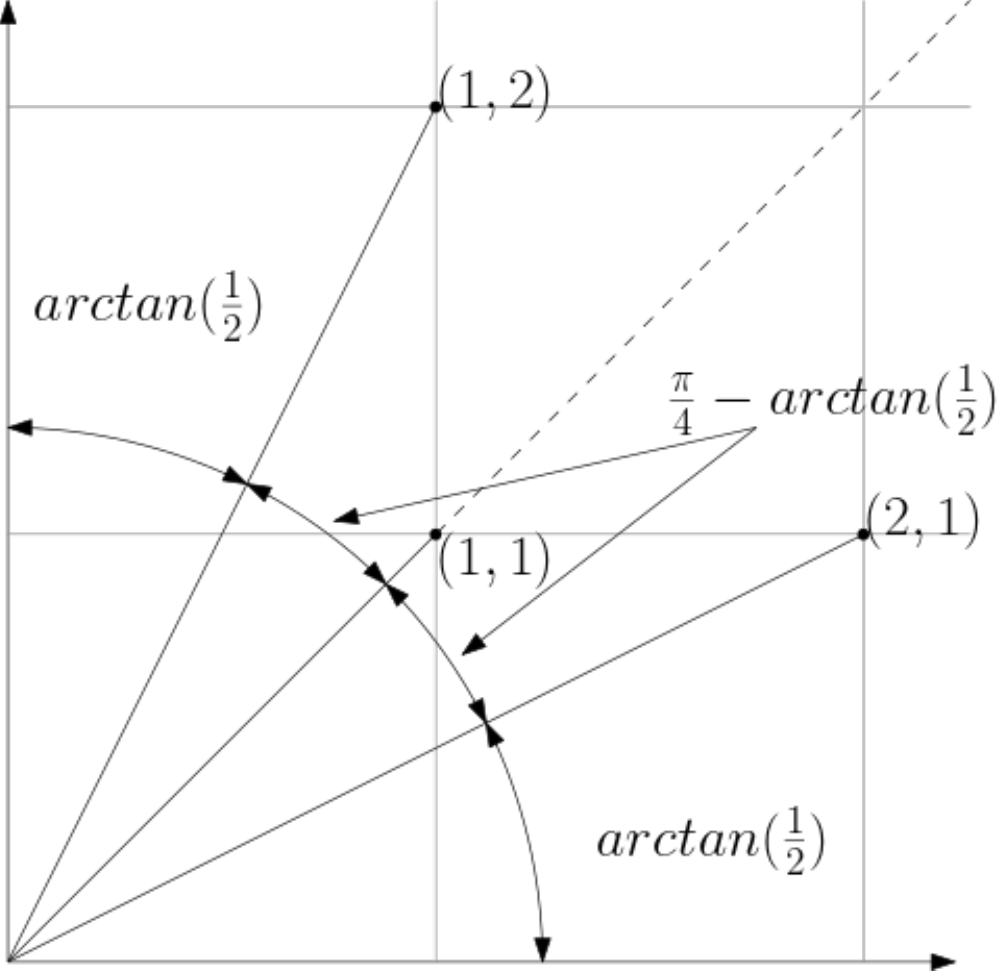}
		\caption{Point, slopes angular sectors used in  \Cref{lemma:distance}.}
		\label{fig:d1}
	\end{minipage}
\end{figure}

\begin{lemma}\label{lemma:distance}
Consider angles $\theta_1$, $\theta_2$ with $0 \leq \theta_1 < \theta_2 \leq \frac{\pi}{2}$ and let $d=\ceil{\frac{1}{\theta_2-\theta_1}}$. Then, a grid point $p$  such that  the edge $e$ that connects the origin $(0,0)$  to $p$ satisfies $\theta_1 < slope(e) < \theta_2$, can be identified as follows:\\

\begin{align*}
\theta_2 - \theta_1 > \frac{\pi}{4}   & \text{\rm ~:~}\quad    p=(1,1)
\\\\
\frac{\pi}{4} \geq  \theta_2 - \theta_1 > arctan(\frac{1}{2}) & \text{\rm ~:~} 
\begin{cases}
p=(1,2) & \text{~if~ } \theta_1 \geq \frac{\pi}{4}\\
p=(1,1) & \text{~if~ } \frac{\pi}{4}> \theta_1 \geq arctan(\frac{1}{2})\\
p=(2,1) & \text{~if~ } arctan(\frac{1}{2}) > \theta_1 
\end{cases}
\\\\
arctan(\frac{1}{2})  \geq  \theta_2 - \theta_1 & \text{\rm ~:~}   
\begin{cases}
p=(d, \floor{tan(\theta_1) \cdot d +1}) & \text{~if~ } \frac{\pi}{4} \geq \theta_2 > \theta_1 \geq 0 \\
p=(1,1) & \text{~if~ } \theta_2 > \frac{\pi}{4}> \theta_1\\
p=(\floor{tan(\frac{\pi}{2} - \theta_2) \cdot d +1},d) & \text{~if~ } \theta_2 > \theta_1 \geq  \frac{\pi}{4}
\end{cases}
\end{align*}

Moreover, if $p=(x,y)$ is the identified point, it also holds that:
\begin{equation*}
\max (x, y)  \leq \frac{\pi}{2} \cdot \frac{1}{\theta_2 - \theta_1}
\end{equation*}
\end{lemma}

\begin{proof}
See \Cref{fig:d1} for points, slopes and angular sectors relevant to \Cref{lemma:distance}.
For each case, we show that the identified points in the statement of the lemma satisfy the ``slope'' (``$\theta_1 < slope(e) < \theta_2$'') and the ``length'' 
(``$\max (x, y) < \ldots $'') conditions.

\begin{description}
\item[Case-1: $\theta_2 - \theta_1 > \frac{\pi}{4}$.]
Point $(1,1)$ is the identified point. In this case, the edge $e$ from the origin $(0,0)$ to $(1,1)$ has slope $\frac{\pi}{4}$.
For the ``slope'' condition, given that $0 \leq \theta_1 < \theta_2 \leq \frac{\pi}{2}$ and $\theta_2 - \theta_1 > \frac{\pi}{4}$, it is enough to show  that
$\theta_1 < \frac{\pi}{4} < \theta_2$ which implies that  $\theta_1 < slope(e) < \theta_2$.
If $\theta_1 > \frac{\pi}{4}$ we have that,  
\begin{align*}
 	\theta_2 - \theta_1 > \frac{\pi}{4} \Leftrightarrow \theta_2 & > \theta_1 + \frac{\pi}{4} \\
	&> \frac{\pi}{4} + \frac{\pi}{4}\\
	&= \frac{\pi}{2}
\end{align*}

A  clear contradiction. So, $\theta_1 < \frac{\pi}{4}$.
In a similar way we can show that $\theta_2 > \frac{\pi}{4}$.

For the ``length'' condition, 
we have to show that, 
\begin{equation*}
max(1,1) = 1 \leq \frac{\pi}{2} \cdot \frac{1}{\theta_2 - \theta_1}
\end{equation*}
This is true since, 
\begin{align*} 
	0 \leq \theta_1 < \theta_2 \leq \frac{\pi}{2} & \Rightarrow  \theta_2-\theta_1  \leq \frac{\pi}{2} \\
	& \Leftrightarrow 	 \frac{1}{\theta_2 - \theta_1}  \geq \frac{1}{\frac{\pi}{2}} \\
	& \Leftrightarrow	 \frac{\pi}{2}\frac{1}{\theta_2 - \theta_1} \geq \frac{\pi}{2} \cdot \frac{1}{\frac{\pi}{2}} = 1
\end{align*}
	
\item[Case-2: $\frac{\pi}{4} \geq  \theta_2 - \theta_1 > arctan(\frac{1}{2})$.]
We first establish the ``slope'' condition. 

For the case where  $arctan(\frac{1}{2}) > \theta_1$, the identified point is $(2,1)$. We note that the slope of the edge $e$ from the origin $(0,0)$ to $(2,1)$ is  $slope(e) = arctan(\frac{1}{2})$.
Then, by the assumption we have: 
\begin{align*}
	\theta_2 - \theta_1 > arctan\left(\frac{1}{2}\right) \Leftrightarrow \theta_2 & > \theta_1 + arctan\left(\frac{1}{2}\right)\\
	& \geq arctan\left(\frac{1}{2}\right)
\end{align*}

It follows that $\theta_1 < arctan(\frac{1}{2}) < \theta_2 \Rightarrow \theta_1 < slope(e) < \theta_2$. 

For the case where  $\frac{\pi}{4}> \theta_1 \geq arctan(\frac{1}{2})$, the identified point is $(1,1)$. We note that the slope of the edge $e$ from the origin $(0,0)$ to $(1,1)$ is $slope(e) = \frac{\pi}{4}$. 
By the assumption, and by taking into account that  $arctan(\frac{1}{2}) > \frac{\pi}{8}$, we have: 
\begin{align*}
\theta_2 - \theta_1 > arctan\left(\frac{1}{2}\right) \Leftrightarrow \theta_2 & > \theta_1 + arctan\left(\frac{1}{2}\right)\\
& \geq arctan\left( \frac{1}{2} \right) + arctan\left(\frac{1}{2}\right)\\
& = 2 \cdot arctan\left(\frac{1}{2}\right)\\
& > 2 \cdot \frac{\pi}{8}\\
& = \frac{\pi}{4}
\end{align*}

It follows that $\theta_1 < \frac{\pi}{4} < \theta_2 \Rightarrow \theta_1 < slope(e) < \theta_2$.

For the case where  $ \theta_1 \geq \frac{\pi}{4}$, the identified point is $(1,2)$. We note that the slope of the edge $e$ from the origin $(0,0)$ to $(1,2)$ is  $slope(e) = arctan(2)$. We want to establish that  $\theta_1 < arctan(2) < \theta_2$. This can be easily proved by taking into account that $arctan(2)= \frac{\pi}{2} - arctan(\frac{1}{2})$ as well as that $2 \cdot arctan(\frac{1}{2}) > \frac{\pi}{4}$.

For the ``length'' condition, it is enough to show that:
\begin{equation*}
max(max(2,1),max(1,1), max(1,2)) = 2 \leq \frac{\pi}{2} \cdot \frac{1}{\theta_2 - \theta_1}
\end{equation*}

This is true since,
\begin{align*}
 \theta_2 - \theta_1 & \leq \frac{\pi}{4}\\
\Leftrightarrow  \frac{1}{\theta_2 - \theta_1} & \geq \frac{1}{\frac{\pi}{4}}\\
\Leftrightarrow \frac{\pi}{2}\frac{1}{\theta_2 - \theta_1} & \geq \frac{\pi}{2} \cdot \frac{1}{\frac{\pi}{4}}\\
 & =2
\end{align*}

\item[Case-3: $arctan(\frac{1}{2})  \geq  \theta_2 - \theta_1$.]
We first establish the ``slope'' condition.
In the case where $\frac{\pi}{4} \geq \theta_2 > \theta_1 \geq 0$, by \Cref{lemma:smallAngle} the identified point immediately satisfies the ``slope'' condition. The same holds for the symmetric case where  $ \theta_2 > \theta_1 \geq \frac{\pi}{4}$. Finally, in the case where  $ \theta_2 > \frac{\pi}{4} > \theta_1 $ the slope condition trivially holds since the edge from the origin $(0,0)$ to $(1,1)$ has slope $\frac{\pi}{4}$.

For the ``length'' condition, we note that in all three cases we have that,
\begin{align*}
\max(x,y) \leq d & = \ceil{\frac{1}{\theta_2-\theta_1}}  \\
& < \frac{1}{\theta_2-\theta_1} +1 
\end{align*}

But it also holds:
\begin{equation}
\frac{1}{x} +1 \leq \frac{\pi}{2} \cdot \frac{1}{x} \text{ , where } 0 < x \leq \frac{\pi}{2} - 1 \label{eq:aab}
\end{equation}

And since $\theta_2 - \theta_1 \leq arctan \left(\frac{1}{2}\right) < \frac{\pi}{2}-1$ we get that,
\begin{equation*}
max(x,y) < \frac{1}{\theta_2-\theta_1} + 1 \overset{\cref{eq:aab}}{\leq} \frac{\pi}{2} \cdot \frac{1}{\theta_2 - \theta_1}
\end{equation*}
The lemma is now proved.
\end{description}
\end{proof}

%% file: algo_1quad.tex
\label{sect:1quad_drawing}

In this Section, we describe an algorithm that builds a monotone drawing of  an $n$-vertex tree  on a grid of size  at most $n \times n$. We refer to this algorithm as ``traditional'' since it satisfies all drawing conventions followed by all algorithms that have appeared in the literature, that is, it takes as input a \emph{rooted ordered tree} $T$ and produces  a monotone drawing of $T$  where (i)~the root of $T$ is drawn at the origin  of the drawing, (ii)~the drawing of $T$ is confined in the first quadrant, and (iii)~the order of the children of each node of $T$  is respected. The algorithm produces a non-strictly slope-disjoint tree drawing which, by \Cref{thm:NSslopeDis_monotonePlanar}, is monotone and planar.

In order to describe a  non-strictly slope-disjoint tree drawing, we need to identify for each vertex $u$ of the tree a grid point to draw $u$ as well as to assign to it two  angles $a_1(u)$, $a_2(u)$, with $a_2(u) > a_1(u)$. For every tree node, the identified grid point and the two angles should be such that the three properties of the non-strictly slope-disjoint drawings are satisfied.

The basic idea behind our algorithm is to split in a balanced way the angle-range $\left< a_1(u), a_2(u)\right>$ of vertex $u$ to its children based on  the size of the subtrees rooted at them. The following strategy formalizes this idea.

\begin{strategy} \label{strategy_main:angle}
Let $u$ be a non-leaf vertex of an $n$-vertex rooted tree $T$ such that we already have  assigned values for $a_1(u)$ and $a_2(u)$, with $a_1(u) < a_2(u)$. Let $v_1,v_2,\ldots,v_m$, $m\geq 1$ be the children of $u$.
We assign angle-ranges for the children of $u$ in the following way:
\begin{align*}
a_1(v_i)  & \text{\rm ~=~}  
\begin{cases}
a_1(u) & \text{~if~ } i=1\\
a_2(v_{i-1})  & \text{~if~ } 1 <i \leq m
\end{cases}
\\
a_2(v_i) & \text{\rm ~=~}    a_1(v_i)+(a_2(u)-a_1(u)) \cdot \frac{|T_{v_i}|}{|T_u|-1},~~   1 \leq i \leq m
\end{align*}
\end{strategy}

The following lemma proves that \Cref{strategy_main:angle} satisfies Property-2 and Property-3 of the non-strictly slope-disjoint drawings.

\begin{lemma} \label{lemma:angle}
Let $u$ be a vertex of the rooted tree $T$ such that we already have  assigned values for $a_1(u)$ and $a_2(u)$, with $a_1(u) < a_2(u)$. Let $v_1, v_2, \ldots, v_m,~m\geq 1$, be the children of $u$ in $T$. If we assign values for angle-ranges of the children of $u$ according to \Cref{strategy_main:angle}, then Property-2 and Property-3 of the non-strictly slope-disjoint drawings are satisfied.
\end{lemma}

\begin{proof}
For Property-3, we have to show that for every $k, l$, $1\leq k<l \leq m$, it holds:
$a_1(v_k)<a_2(v_k) \leq a_1(v_l) < a_2(v_l)$.
For any $j,~ 1\leq j \leq m$, we have that, 
\begin{align*}
a_2(v_j)   	& =   a_1(v_j)+ (a_2(u)-a_1(u)) \cdot \frac{|T_{v_j}|}{|T_u|-1} \\
			& >  a_1(v_j)
\end{align*}

The last inequality holds since,  by assumption,  
$a_2(u) - a_1(u) > 0$ and because the size of a rooted tree is always positive.
Therefore,
\begin{align*}
a_1(v_1) 	& < a_2(v_1) \\
			& = a_1(v_2)  \\
			& < a_2(v_2)  \\
			& = a_1(v_3)  \\
			& \vdots      \\
			& <a_2(v_{m-1})\\
			& =a_1(v_m)\\
			& <a_2(v_m)
\end{align*}

So, for any $k,~l$, $1\leq k<l \leq m$, it holds that $a_1(v_k)<a_2(v_k) \leq a_1(v_l) < a_2(v_l)$ and, thus,  Property-3 holds.

For Property-2, since we proved that $a_1(v_1)<a_2(v_1)\leq a_1(v_2) < \ldots < a_2(v_{m-1}) \leq a_1(v_m)<a_2(v_m)$, it is sufficient to show that $a_1(u) \leq a_1(v_1)$ and $a_2(v_m) \leq a_2(u)$.
The first part trivially holds since $a_1(v_1)=a_1(u)$ by definition.
For the  second part, by using repeatedly the assignment for $a_1$ and $a_2$ provided in the statement of the lemma we get that,
\begin{equation*}
a_2(v_m) = a_1(v_1)+ (a_2(u)-a_1(u))\frac{\sum_{i=1}^{m}|T_{v_i}|}{|T_u|-1}
\end{equation*}
Since  the subtree rooted at $u$, consists of the root vertex $u$ and the subtrees rooted at $u$'s  children, it holds that $|T_u| = \sum_{i=1}^{n}|T_{v_i}| + 1$.
It follows that $a_2(v_m) = a_2(u)$ and Property-2 is satisfied.
\end{proof}

\begin{observation}\label{obs:path}
If a vertex $u$ has only one child, say $v_1$, then the angle assignment \Cref{strategy_main:angle} assigns 
$a_1(v_1)=a_1(u)$ and
$a_2(v_1)=a_2(u)$, which means that the child ``inherits'' the angle-range of its parent.
\end{observation}

\Cref{algo:1_quadrant} describes our monotone tree drawing algorithm. It consists of three steps: Procedure \textsc{AssignAngles} which assigns angle-ranges to the vertices of the tree according to \Cref{strategy_main:angle}, Procedure \textsc{DrawVertices} which assigns each tree vertex to a grid point according to \Cref{lemma:distance} and Procedure \textsc{BalancedTreeMonotoneDraw} which assigns the root to point $(0,0)$ with angle-range $\left<0, \frac{\pi}{2}\right>$ and initiates the drawing of the tree. 

\begin{algorithm}[ht]
\caption{One-Quadrant Monotone Rooted Ordered Tree Drawing}\label{algo:1_quadrant}
\begin{algorithmic}[1]
\Procedure{BalancedTreeMonotoneDraw}{}
\State  \hspace*{-0.5cm}Input: An $n$-vertex  tree $T$ rooted at vertex $r$.
\State  \hspace*{-0.5cm}Output: A monotone drawing of $T$ on a grid of size  at most $n \times n$.
\State $a_1(r) \gets 0, ~a_2(r) \gets \frac{\pi}{2}$
\State \Call {AssignAngles}{$r, ~a_1(r), ~a_2(r)$}
\State Draw $r$ at $(0,0)$
\State \Call {DrawVertices}{$r$}
\EndProcedure

\State
\Procedure{AssignAngles}{u, $a_1$, $a_2$}
\State  \hspace*{-0.5cm}Input: A vertex $u$ and the boundaries of the angle-range $\left<a_1, a_2\right>$ assigned to $u$.
\State  \hspace*{-0.5cm}Action: It assigns angle-ranges to the vertices of $T_u$.
	\For {each child $v_i$ of $u$}
	\State Assign $a_1(v_i),~a_2(v_i)$ as described in \Cref{strategy_main:angle}.
	\State \Call {AssignAngles}{$v_i, ~a_1(v_i), ~a_2(v_i)$}
	\EndFor
\EndProcedure

\State
\Procedure{DrawVertices}{u}
\State \hspace*{-0.5cm}Input: A vertex $u$ where $u$ has already been drawn on the grid and angle-ranges have been defined for all vertices of $T_u$. 
\State \hspace*{-0.5cm}Action: It draws  the vertices of  $T_u$.

	\For {each child $v_i$ of $u$}
	\State Find a valid pair $(x,y)$  as described in \Cref{lemma:distance} where 
	\State $~~~~~\theta_1 \gets a_1(v_i) \text{ and } \theta_2 \gets a_2(v_i)$
	\State If $u$ is drawn at $(u_x, u_y)$, draw $v_i$ at $(u_x+x,u_y+y)$
	\State \Call {DrawVertices}{$v_i$}
	\EndFor
\EndProcedure
\end{algorithmic}
\end{algorithm}

\begin{lemma}\label{lemma:algMonotone}
The drawing produced by \Cref{algo:1_quadrant} is monotone and planar.
\end{lemma}

\begin{proof}
The angle-range assignment of \Cref{strategy_main:angle}  satisfies  Property-2 and Property-3 of the non-strictly slope disjoint drawing as proved in \Cref{lemma:angle}.
In addition, the assignment of the vertices to grid points satisfies  Property-1 of   the non-strictly slope disjoint drawing as proved in \Cref{lemma:distance}.
Thus, the produced drawing by  \Cref{algo:1_quadrant} is non-strictly slope disjoint and, by \Cref{thm:NSslopeDis_monotonePlanar}, it is monotone and planar.
\end{proof}

It remains to establish a bound on the grid size required by \Cref{algo:1_quadrant}. Our proof uses induction on the number of tree vertices having at least two children.

\begin{lemma} \label{lemma:bound}
Let $T$ be a rooted tree and  $u \in T$ be a vertex.
Let $\phi_u = a_2(u)-a_1(u)$ where $a_1(u)$ and $a_2(u)$ are  assigned by \Cref{algo:1_quadrant}. Then, the side-length of the grid  which \Cref{algo:1_quadrant} uses for the  drawing of the subtree $T_u$ rooted at $u$  is bounded by: 
\begin{equation*}
(|T_u|-1)  \frac{\pi}{2}  \frac{1}{\phi_u}
\end{equation*}
\end{lemma}

\begin{proof}
We use induction on the number of vertices having at least two children.
Let $i$ be the number of vertices of $\in T_u$ with at least two children.
\begin{description}
\item[Base Case (i=0):] In this case, $T_u$ is just a path and, by \Cref{obs:path}, \Cref{algo:1_quadrant} assigns to every vertex the same angle-range. 
From this observation, for any vertex $v \in T_u$, it holds that $a_2(v)-a_1(v)=a_2(u)-a_1(u)=\phi_u$ and, therefore, by \Cref{lemma:distance} we have that each edge expands our grid at most by:

\begin{equation*}
\frac{\pi}{2}  \frac{1}{\phi_u}
\end{equation*}

Since the tree has $|T_u|$ vertices, we expand the grid $|T_u|-1$ times, therefore the side-length of the grid required for the drawing of tree $T_u$ is:
\begin{equation*}
(|T_u|-1)  \frac{\pi}{2}  \frac{1}{\phi_u}
\end{equation*}

The base case is now settled.

\item[Induction Step:] We assume that for any rooted subtree which contains at most $i$ vertices  with at least two children each, the statement holds.
We prove that for any  subtree rooted at vertex $u$ with $i+1$ vertices in $T_u$ having at least two children each, the statement also holds.

At first we prove that the only case of interest is when the subtree is rooted at a vertex with at least two children.
Let's assume $T_u$ is the union of a path starting from $u$ and ending at $v$ where each vertex has exactly one child except $v$ and the subtree rooted at $v$.
The number of vertices in $T_v$ having at least two children is $i+1$ by assumption since the vertices in the path between $u$ and $v$ have exactly one child.
If the statement holds for $v$ we have, by \Cref{obs:path},  $a_2(v)=a_2(u)$ and $a_1(v)=a_1(u)$, and thus, 
\begin{equation}\label{eq:b1}
\phi_v = a_2(v)-a_1(v)=a_2(u)-a_1(u)=\phi_u
\end{equation}

The side-length of the required grid for $T_v$ is,
\begin{equation*}
(|T_v|-1)  \frac{\pi}{2}  \frac{1}{\phi_v} \overset{\left(\ref{eq:b1}\right)} {=}  (|T_v|-1)  \frac{\pi}{2}  \frac{1}{\phi_u}
\end{equation*}

Also,  the side-length of grid  required for the path from $u$ to $v$, having   $|T_u| - |T_v|$ edges, is

\begin{equation*}
(|T_u|-|T_v|)  \frac{\pi}{2}  \frac{1}{\phi_u}
\end{equation*}

So,  the total side-length of the   required grid is:
\begin{equation*}
(|T_u|-1)  \frac{\pi}{2}  \frac{1}{\phi_u}
\end{equation*}
Therefore, it is enough to  only consider the  case where the root $u$ of the subtree $T_u$ has at least two children.

Let $u$ be a vertex $\in T$ such that $u$ has at least two children and $T_u$ has $i+1$ vertices with at least two children each.
Let $v_1,v_2,\ldots,v_m$ be the children of $u$ and observe that 
the largest grid devoted to any of the trees\footnote{Recall that  by $T_v^u$, where $v$ is a child of $u$, we denote the tree that consists of edge $(u,v)$ and $T_v$.}
$T_{v_j}^u, 1\leq j \leq m$, 
determines the side-legth of the grid  drawing of $T_u$ since the subtrees rooted at the children of $u$ are drawn completely inside non-overlapping (but possibly touching) angular sectors. 
The above statement holds because all the grids that are used for the subtrees have the same origin ($u$) and all angular sectors lies in the first quadrant since \Cref{algo:1_quadrant} assigns to the root  angle-range $\left<0, \frac{\pi}{2}\right>$.
Therefore, the   side-length of the grid required in order to  draw $T_u$ is equal to the maximum of the grid side-lengths required to draw any of $T_{v_j}^u$.
For any vertex $v_j$, since vertex $u$ has at least two children, it holds that the number of vertices in $T_{v_j}$ having at least two children each is less or equal to $i$, and therefore the  induction hypothesis applies.
Thus, $T_{v_j}$  is drawn on a grid with side-length  bounded by,

\begin{equation*}
(|T_{v_j}|-1)  \frac{\pi}{2}  \frac{1}{\phi_{v_j}}
\end{equation*}

For the edge connecting $u$ to $v_j$, by \Cref{lemma:distance} we require a grid of side-length  bounded by,

\begin{equation*}
\frac{\pi}{2}  \frac{1}{\phi_{v_j}}
\end{equation*}

Therefore, the required  grid  has total side-length  bounded by:

\begin{equation*}
|T_{v_j}|  \frac{\pi}{2}  \frac{1}{\phi_{v_j}}
\end{equation*}

Since we applied \Cref{strategy_main:angle}, it holds that 
\begin{equation}
\phi_{v_j}= \frac{|T_{v_j}|}{|T_u|-1}  \phi_u \label{eq:ac2}
\end{equation}

Thus, the required total grid side-length required  can  be restated as:
\begin{align*}
|T_{v_j}|  \frac{\pi}{2}  \frac{1}{\phi_{v_j}} &
\overset{\left(\ref{eq:ac2}\right)}{=} |T_{v_j}|  \frac{\pi}{2}  \frac{1}{ \frac{|T_{v_j}|}{|T_u|-1}  \phi_u }\\
& = (|T_u|-1)  \frac{\pi}{2}  \frac{1}{\phi_{u}}
\end{align*}

Therefore, the statement holds for the induction step. This completes the proof of the lemma.
\end{description}
\end{proof}

\begin{theorem} \label{thm:1Q_algo_gridSize}
Given a rooted n-vertex Tree $T$, \Cref{algo:1_quadrant} produces a  monotone planar grid drawing using a   grid of size  at most $n \times n$.
\end{theorem}

\begin{proof}
The monotonicity and  planarity of the drawing follow directly from \Cref{lemma:algMonotone}. 
\Cref{algo:1_quadrant}  assigns to the root $r$ of  tree $T$  the  angle-ranger $< 0, \frac{\pi}{2}>$, i.e., $a_1 (r)= 0$ and $a_2(r) = \frac{\pi}{2}$.
By applying \Cref{lemma:bound} to the root   $r$  we get that, in the worst case, the drawing of $T$ uses a grid of side-length that is smaller  or equal to:
\begin{equation*}
\left( n-1 \right)  \frac{\pi}{2}  \frac{1}{\frac{\pi}{2}} = n-1
\end{equation*}

Therefore, the required grid is of size at most $n \times n$.
\end{proof}

\Crefrange{fig:1Q_algo_binTree}{fig:1Q_algo_worstTree}  present  drawings produced by  \Cref{algo:1_quadrant}. \Cref{fig:1Q_algo_binTree} shows the drawing of a 5-layer complete binary tree (31 vertices). While \Cref{thm:1Q_algo_gridSize} indicates that a grid of size $31 \times 31$ may be required, the binary tree is drawn on a $23 \times 22 $ grid. 
\Cref{fig:1Q_algo_path} shows the drawing of a path (15 vertices). All paths that are rooted at one of their endpoints are drawn along the main diagonal of a grid with side-length matching the bound stated in \Cref{thm:1Q_algo_gridSize}. 
Finally, \Cref{fig:1Q_algo_worstTree} shows a drawing of a tree (out of all 10-vertex rooted trees) that requires maximum area (when produced by \Cref{algo:1_quadrant}). We have drawn all 10-vertex rooted trees and have identified non-path trees that require  a grid  of the dimensions stated in \Cref{thm:1Q_algo_gridSize}.

\begin{figure}[h]
	\centering
	\begin{minipage}[t]{0.30\textwidth}
		\centering
		\includegraphics[scale=0.25]{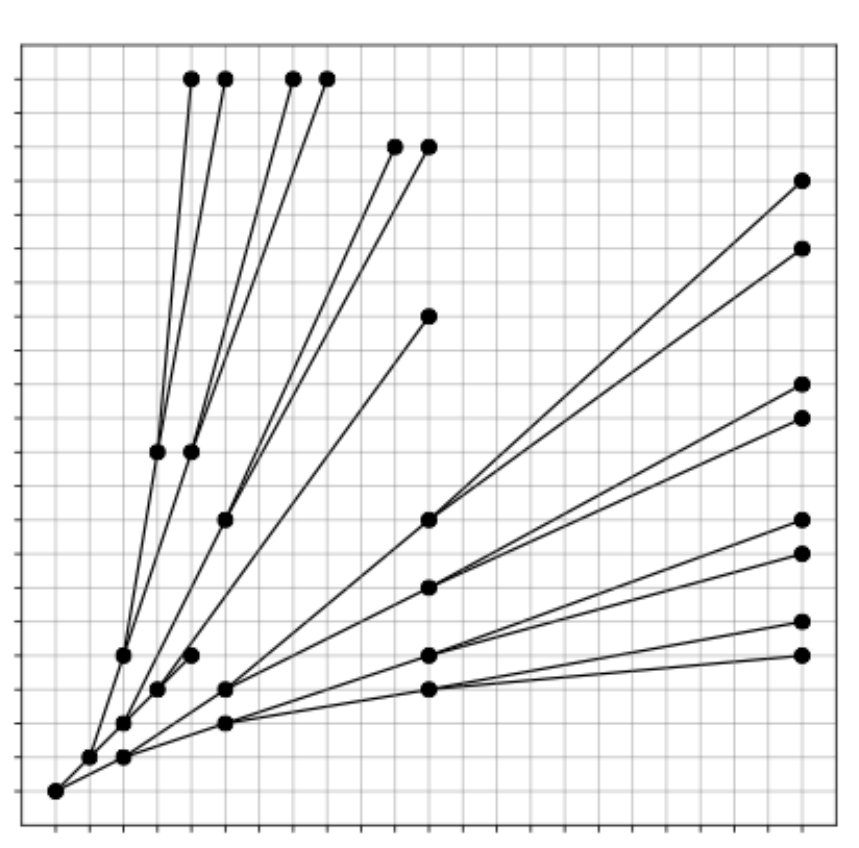}
		\caption{A full binary tree (31 vertices) as drawn by \Cref{algo:1_quadrant}. Grid size: $23 \times 22$.}
		\label{fig:1Q_algo_binTree}
	\end{minipage}
\hfill
	\begin{minipage}[t]{0.30\textwidth}
		\centering
		\includegraphics[scale=0.25]{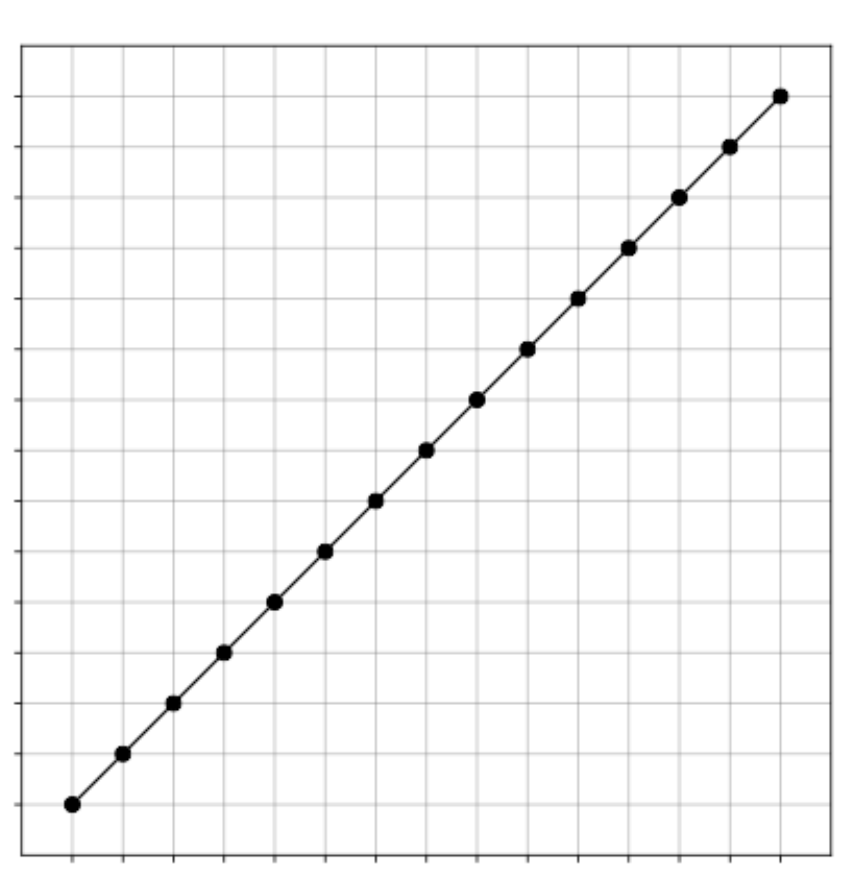}
		\caption{ A path (15 vertices) as drawn by \Cref{algo:1_quadrant}. Grid size: $15 \times 15$.}
		\label{fig:1Q_algo_path}
	\end{minipage}
\hfill
	\begin{minipage}[t]{0.30\textwidth}
		\centering
		\includegraphics[scale=0.25]{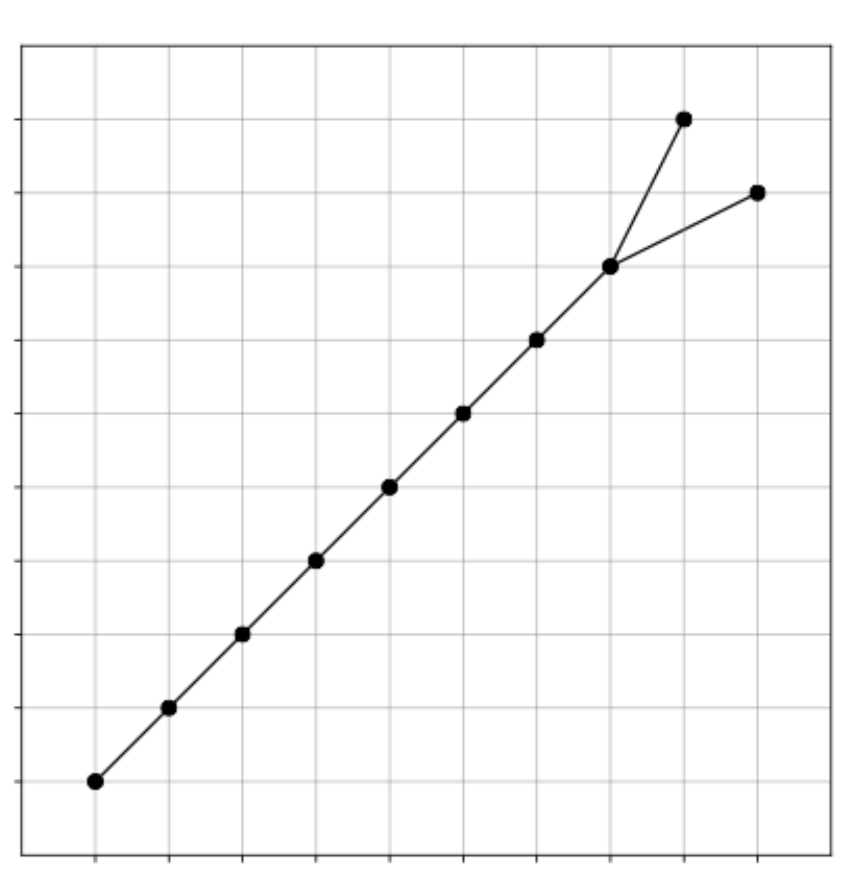}
		\caption{ A non-path tree (10 vertices) with maximum required area when drawn  by \Cref{algo:1_quadrant}. Grid size: $10 \times 10$.}
		\label{fig:1Q_algo_worstTree}
	\end{minipage}
\end{figure}

%% file: algo_1quad_convex.tex
\label{sec:convex_drawings}
In this section, we focus on convex monotone tree drawings. Recall that a tree drawing is \emph{convex} if, by extending every edge incident to a leaf into an infinite ray (originating at its parent and passing through the leaf), the resulting set of rays do not intersect and they (together with the original tree edges) partition the plane into convex regions. As it can be seen in \cref{fig:non_convex_by_1quad}, there exists ordered trees for which   \cref{algo:1_quadrant} generates non-convex drawings. \cref{fig:convex_by_conv_alg} shows, for the same tree,  the convex monotone drawing generated by \cref{algo:convex} (developed in this section).

\begin{figure}[h]
	\centering
	\begin{minipage}[t]{0.45\textwidth}
		\centering
		\includegraphics[scale=0.2]{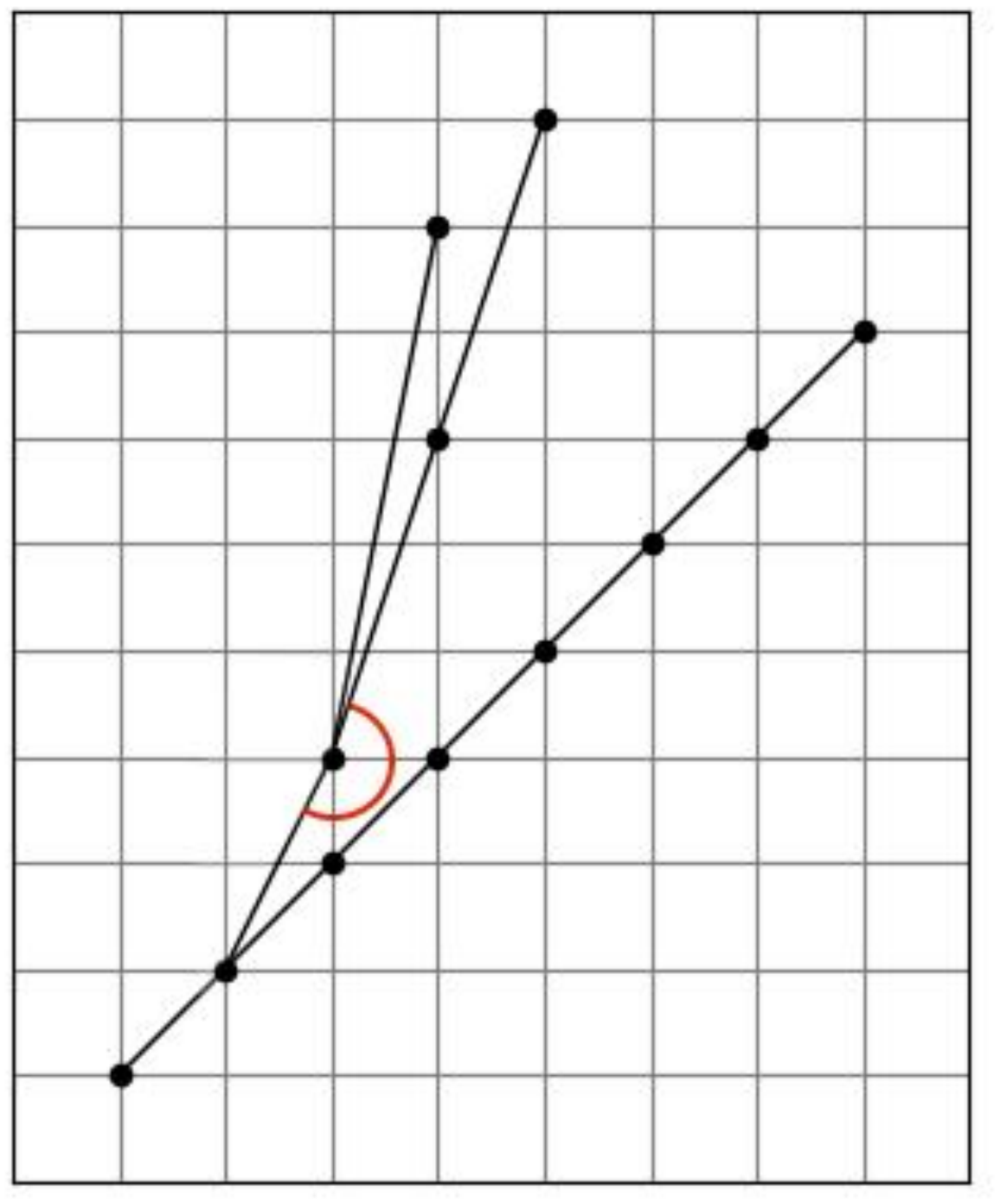}
		\caption{Non-convex monotone drawing of a small tree (12 vertices) by~\cref{algo:1_quadrant}: the indicated angle at the central fork-highlighted via a red arc-is  non-convex.}
		\label{fig:non_convex_by_1quad}
	\end{minipage}
	\hfill
	\begin{minipage}[t]{0.45\textwidth}
		\centering
		\includegraphics[scale=0.2]{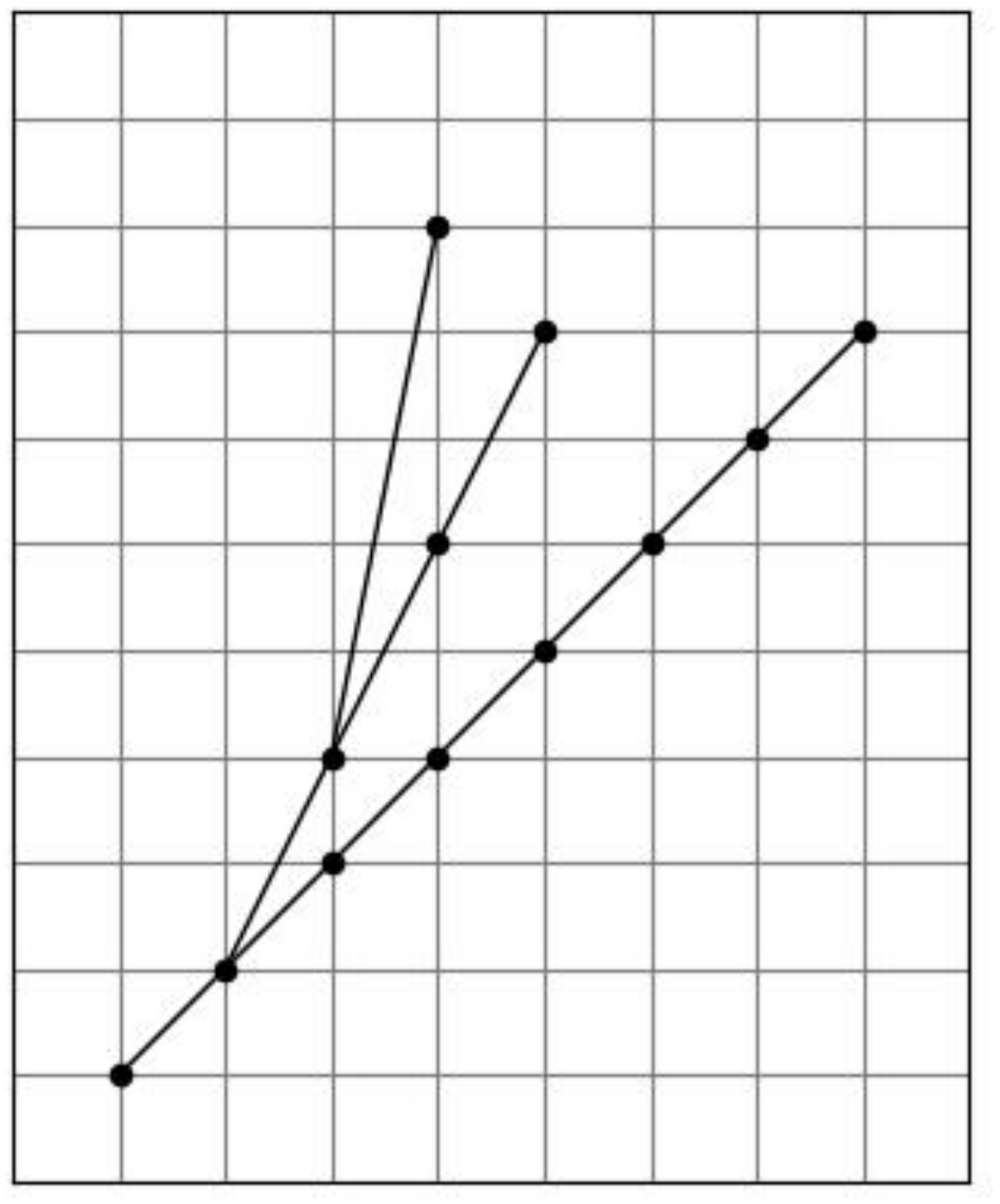}
		\caption{Convex monotone drawing of the tree in \cref{fig:non_convex_by_1quad}  by~\cref{algo:convex}: slope inheritance straightens the central fork angle to exactly~$\pi$.}
		\label{fig:convex_by_conv_alg}
	\end{minipage}
\end{figure}

We first illustrate the two cases in which~\cref{algo:1_quadrant} can produce a non-convex drawing, corresponding to non-convex angles incident to (i) the root (\cref{fig:bad_case_one}) and (ii) a non-root vertex (\cref{fig:bad_case_two}).

\begin{itemize}[align=parleft,left=0pt..1em]
\item\emph{Case-1: Non-convex angle at the root (\cref{fig:bad_case_one})}
Consider the tree consisting of a root $r$ with exactly two children $v_1$ and $v_2$.  \cref{algo:1_quadrant} assigns $r$ the sector $\langle0,\tfrac\pi2\rangle$ and, by \cref{strategy_main:angle}, places $v_1$ and $v_2$ at grid points whose edge slopes satisfy $0 < \alpha_1 < \alpha_2 < \frac\pi2.$
Although both edges lie in the first quadrant, the outer face incident to $r$ spans the complement of that quadrant.  Hence the interior angle at $r$ is $2\pi - (\alpha_2 - \alpha_1) > \pi,$ which is non-convex.

\item \emph{Case-2: Non-convex angle at a non-root vertex (\cref{fig:bad_case_two})}
Now consider a long-path on \(n\) vertices in which the second vertex has an extra child.  When placing these vertices, \cref{lemma:distance} can produce a discontinuous change in the assigned slope whenever $\theta_2 - \theta_1$ lies near an integer multiple of the grid‐alignment increment (cf.\  \cref{lemma:distance}).  In practice, shrinking the angular sector allocated to the extra child by an arbitrarily small amount may force the subsequent path‐edge’s slope to “jump” to a different grid point.  As a result, the angle between that edge and its predecessor along the path exceeds $\pi$, creating a non-convex angle at the intermediate node (see  \cref{fig:bad_case_two}).
\end{itemize}

In both cases, the drawings produced by \cref{algo:1_quadrant}, while monotone, are not convex. These examples motivate the adjustments in  \cref{algo:convex} to remedy these issues by selectively re-using edge slopes when possible and otherwise invoking  \cref{lemma:distance} to place children so that all incident angles remain strictly less than $\pi$. We summarize these key differences in \cref{obs:changes_to_convex}, and prove the main properties about \cref{algo:convex} in \cref{thm:convex_drawing}.

\begin{figure}[h]
	\centering
	\begin{minipage}[t]{0.45\textwidth}
		\centering
		\includegraphics[scale=0.35]{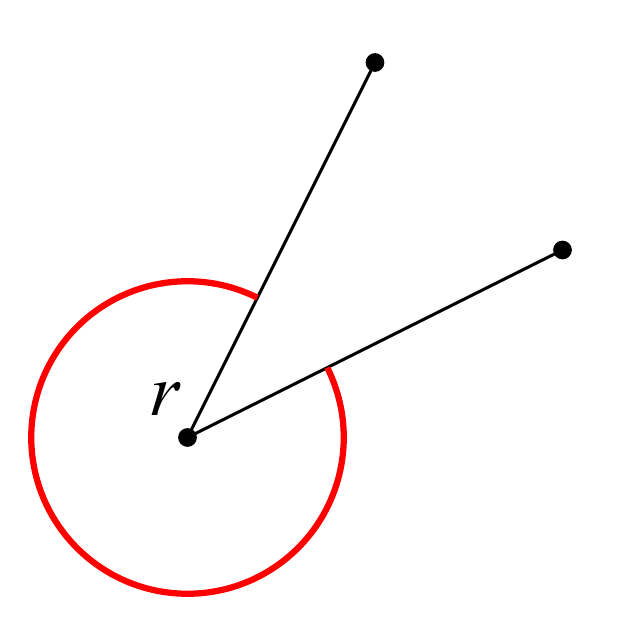}
		\caption{\cref{algo:1_quadrant} places the two leaves in the first quadrant, so the face between them wraps around the origin and forms a non-convex angle at root $r$ (red arc), resulting in a non-convex drawing.}
		\label{fig:bad_case_one}
	\end{minipage}
	\hfill
	\begin{minipage}[t]{0.45\textwidth}
		\centering
		\includegraphics[scale=.35]{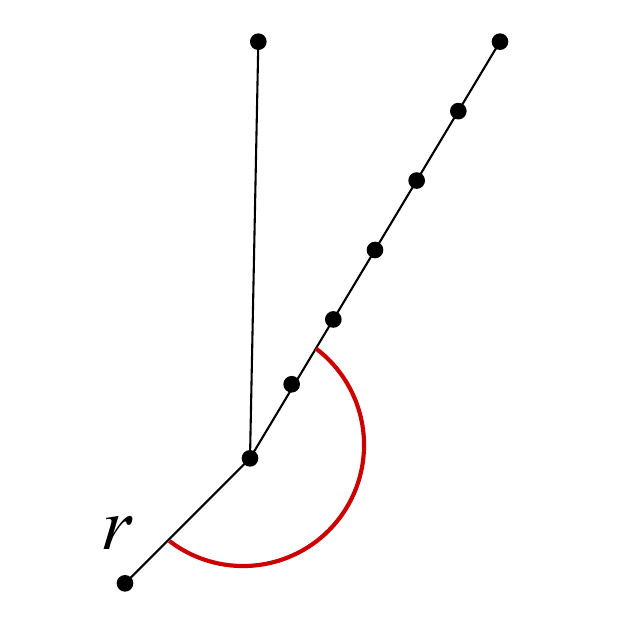}
		\caption{The angle formed at the fork between two consecutive edges of the long path and the short path at the central node (red arc) can be strictly non-convex. 
		}
		\label{fig:bad_case_two}
	\end{minipage}
\end{figure}

\begin{algorithm}[ht]
\caption{Convex Tree Drawing}\label{algo:convex}
\begin{algorithmic}[1]
\Procedure{ConvexDraw}{}
    \State {Input:} An $n$-vertex tree $T$.
    \State {Output:} A convex drawing of $T$ on a grid of size at most $n \times n$.
    \State Root $T$ at any vertex of degree $1$ (such a vertex always exists in a tree).
    \State Set $a_1(r) \gets 0$ and $a_2(r) \gets \frac{\pi}{2}$.
    \State \Call{AssignAngles}{$r,\, a_1(r),\, a_2(r)$} (see \cref{algo:1_quadrant}).
    \State Draw $r$ at $(0,0)$.
    \State \Call{DrawConvexVertices}{$r,\, \emptyset$}.
\EndProcedure

\State
\Procedure{DrawConvexVertices}{$u,\, (\delta_x,\delta_y)$}
    \State {Input:} A vertex $u$ drawn at $(u_x,u_y)$ and (if $u$ is not the root) the displacement vector $(\delta_x,\delta_y)$ from its parent to $u$.
    \State {Action:} Draw the vertices in the subtree $T_u$.
    \For {each child $v_i$ of $u$}
        \If {$slope((0,0), (\delta_x,\delta_y) ) \in (a_1(v_i),a_2(v_i))$} \Comment{We examine the slope between $u$'s parent and $u$.}
            \State $x \gets \delta_x$, and $y \gets \delta_y$
        \Else
            \State Determine a valid pair $(x,y)$ as specified in \cref{lemma:distance} with $\theta_1 \gets a_1(v_i)$ and $\theta_2 \gets a_2(v_i)$.
        \EndIf
        \State If $u$ is drawn at $(u_x, u_y)$, then draw $v_i$ at $(u_x+x,\,u_y+y)$.
        \State \Call{DrawConvexVertices}{$v_i,(x,y)$}.
    \EndFor
\EndProcedure
\end{algorithmic}
\end{algorithm}

\notshow{
\begin{observation}\label{obs:changes_to_convex}
The modifications in \cref{algo:convex} relative to \cref{algo:1_quadrant} are:

\begin{enumerate}
    \item \textbf{Slope Verification and Adjustment.} For every non-leaf vertex $u$ and each child $v$ of $u$, the algorithm checks whether the direction (slope) of the edge from $u$, as inherited from the displacement vector $(\delta_x,\delta_y)$ 
    of the edge incidental at $u$, lies within the open interval $(a_1(v),a_2(v))$. If this condition holds, the child is placed using the same displacement. Otherwise, a displacement vector is computed (using \cref{lemma:distance}) similar to \cref{algo:1_quadrant}.
    \todo[inline] {Ant. This term "displacement vector" appears for the first time. We should introduce it or use a different description. I have to admit that I made several passes and I did not managed to follow the description. it is not clear from this description that one child inherits the slope of the "from the parent" incoming edge.  Should we give a simple example with say, binary trees, where this problem is created by alg-1?}
    \item \textbf{Rooting at a Leaf.} The tree is rooted at a vertex of degree $1$, which implies that the unique edge incident to the root is drawn with slope $\frac{\pi}{4}$.
\end{enumerate}
Both modifications are illustrated in \cref{fig:convex tree 2}. In that figure, the tree is rooted at a leaf, and for every non-leaf node at least one child is drawn using an edge whose slope is inherited from it's parent.
\end{observation}
}

\begin{observation}\label{obs:changes_to_convex}
The main changes in \cref{algo:convex} relative to \cref{algo:1_quadrant} are twofold: (i) \cref{algo:convex} explicitly verifies whether the angular sector assigned to a child vertex contains the slope of the edge from its parent and adjusts the drawing accordingly to ensure convexity; and (ii) the tree is rooted at a node with degree $1$ rather than at an arbitrary vertex. 
\end{observation}

The  modifications mentioned in \cref{obs:changes_to_convex} are clearly illustrated in \cref{fig:convex tree 2,fig:convex tree 1}.
In each figure, the tree is rooted at a vertex of degree $1$, and for every non-leaf node, at least one child is drawn with an edge that maintains the same slope as that of its parent. 

\begin{figure}[h]
	\centering
	\begin{minipage}[t]{0.45\textwidth}
		\centering
		\includegraphics[scale=0.4]{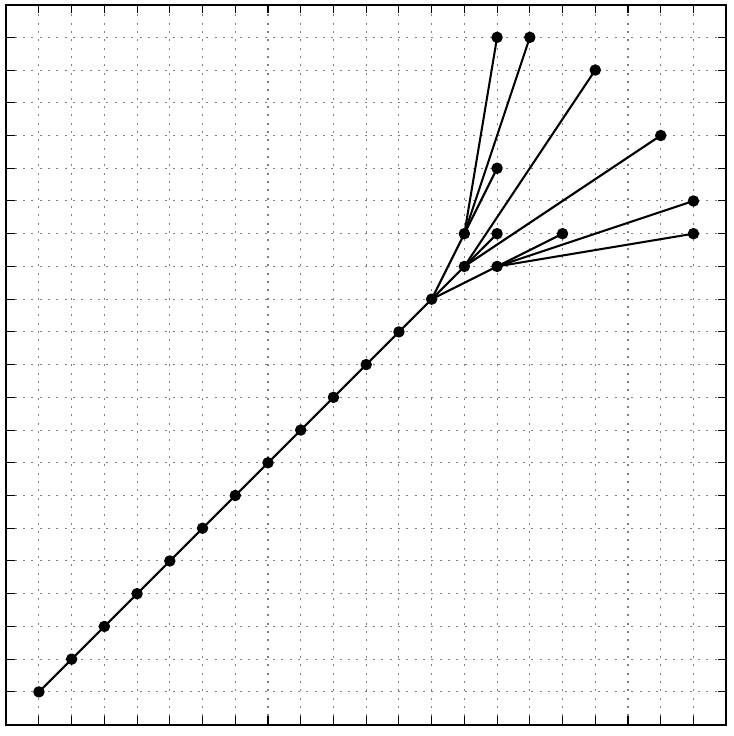}
		\caption{A convex monotone drawing of a broomstick-like tree (25 vertices) as drawn by \cref{algo:convex}. Grid size: $21 \times 21$.}
		\label{fig:convex tree 1}
	\end{minipage}
	\hfill
	\begin{minipage}[t]{0.45\textwidth}
		\centering
		\includegraphics[scale=0.4]{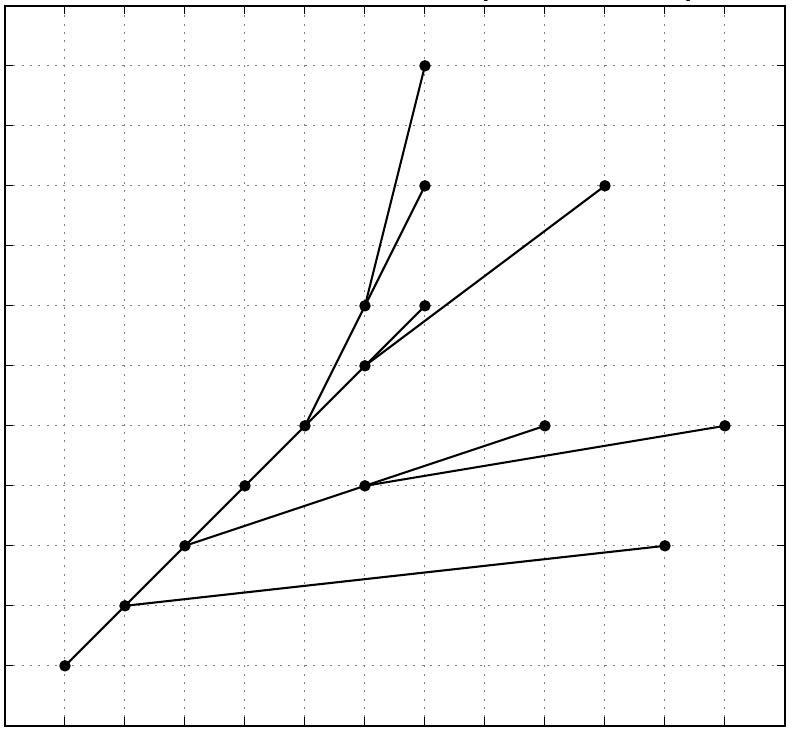}
		\caption{A convex monotone drawing of a tree with a small degree per node (15 vertices) as drawn by \cref{algo:convex}. Grid size: $12 \times 11$.}
		\label{fig:convex tree 2}
	\end{minipage}
\end{figure}

\begin{theorem}\label{thm:convex_drawing}
The tree drawing produced by \cref{algo:convex} is convex monotone and fits in a grid of size at most $n \times n$.
\end{theorem}

\begin{proof}
We prove the theorem in two parts. First, we show that the drawing is planar, non-strictly slope disjoint, and fits in a grid of size at most $n \times n$. Second, we prove that the drawing is convex. For the first part, we provide only a sketch, since the arguments closely mirror those presented in \cref{lemma:algMonotone}, \cref{lemma:bound}, and \cref{thm:1Q_algo_gridSize}.

\textbf{Grid Size \& Monotonicity [Sketch]:} We first sketch that the drawing is non-strictly slope disjoint. \cref{algo:convex} assigns each vertex an angle‐range and places its children so that every edge’s slope lies strictly between the lower and upper bounds of that range. This construction is analogous to the one used in \cref{lemma:algMonotone}. Consequently, the drawing is non‑strictly slope‑disjoint. By \cref{thm:NSslopeDis_monotonePlanar}, any drawing that is non‑strictly slope‑disjoint is guaranteed to be planar and monotone.

In high-level, the proof for the grid size follows because children inherit a sufficiently large angular sector from their parent, the grid size bound proof follows similarly to the previous argument. Once the root is fixed, the analysis for the grid size bound follows essentially as in the one‐quadrant monotone drawing in \cref{thm:1Q_algo_gridSize}. The key observation is that in \cref{algo:convex} the only difference from \cref{algo:1_quadrant} that in some cases a child inherits the displacement (and hence the slope) from its parent. In these cases the child’s angular sector is at most the parent’s, that is, for a parent $u$ and its child $v$ we have
$$a_2(v) - a_1(v) \leq a_2(u) - a_2(u).$$
This means that the expansion factor contributed by the edge from $u$ to $v$ (which is bounded by a term proportional to $1/(a_2(v) - a_2(v))$ as shown in \cref{lemma:bound}) does not worsen. Therefore, by applying an inductive argument similar to that in the proof of \cref{lemma:bound}, and following the technique in \cref{thm:1Q_algo_gridSize} we conclude that the drawing fits in an $n\times n$ grid.

\textbf{Convexity:} Recall that a drawing is convex if (i) at every vertex with degree at least two the interior angles between consecutive edges are convex, and (ii) if each edge incident to a leaf is extended into an infinite ray (originating at its parent and passing through the leaf), then these do not intersect.

We consider an arbitrary vertex $u$ with degree at least two. (The root is excluded since it has degree $1$ by construction.) Let $u$ be the parent of $v$, and let the edge $(u,v)$ have slope $\alpha$. 
We take two cases based on whether there is a child $v'$ of $v$ whose angular sector strictly contains the slope $\alpha$ (i.e.\ $\alpha\in(a_1(v'),a_2(v'))$).  By Properties~2 and 3 of non‑strictly slope‑disjoint drawings, there can be at most one such child.  If no child’s sector contains $\alpha$ in its interior, then again by Properties~2 and 3 the slope $\alpha$ can lie on the boundary of at most two children’s sectors: namely there exist children $v_1$ and $v_2$ of $v$ with $a_2(v_1)=\alpha = a_1(v_2).$ Since \cref{algo:convex} partitions a node’s angular sector into non‑overlapping sectors assigned to its children, it follows that exactly those two children will have $\alpha$ lying on the boundary of their respective sectors.

\textbf{Case 1.} \emph{There exists a child $v'$ of $v$ such that $\alpha \in (a_1(v'),a_2(v'))$.}\\[1mm]
In this case,  \cref{algo:convex} places $v'$ so that 
\[
slope(v,v') = \alpha.
\]
Thus, the edge $(v,v')$ is collinear with $(u,v)$ (or can be regarded as a continuation of it), and the angle between these two edges is $\pi$. Regardless of the other incident edges at $v$, the angles between consecutive edges at $v$ remain convex.

\textbf{Case 2.} \emph{No child $v'$ of $v$ satisfies $\alpha \in (a_1(v'),a_2(v'))$.}\\[1mm]
In this case, the angular range $[a_1(v),a_2(v)]$ at $v$ is partitioned among the children such that there exist exactly two children, say $v_1$ and $v_2$, with $$a_2(v_1) = a_1(v_2)=\alpha$$ and also
\[
[a_1(v_1),a_2(v_1)] \subseteq [a_1(v),\alpha] \quad \text{and} \quad [a_1(v_2),a_2(v_2)] \subseteq [\alpha,a_2(v)].
\]
By the non-strict slope-disjoint property (Property~1), the slopes satisfy
\[
a_1(v) < slope(v,v_1) < \alpha < slope(v,v_2) < a_2(v).
\]
Thus, the angles between $(u,v)$ and $(v,v_1)$ and between $(u,v)$ and $(v,v_2)$ are both strictly less than $\pi$. Moreover, since all slopes lie in the interval $\left[0,\frac{\pi}{2}\right]$, the angle between $(v,v_1)$ and $(v,v_2)$ is also convex. Hence, every angle at $v$ formed by two consecutive incident edges is convex.

\textbf{Non-Intersection of Extended Leaf Rays:}
Next, we show that extending each leaf into an infinite ray results in a set of non-intersecting rays. Since the root $r$ has degree $1$, let $v$ be its unique child. By the initialization in \cref{algo:convex}, we have $a_1(r)=0$ and $a_2(r)=\frac{\pi}{2}$, so by \cref{lemma:distance} the edge $(r,v)$ is drawn with a slope of $\frac{\pi}{4}$ (and consequently, the reverse direction from $v$ to $r$ has slope $-\frac{\pi}{4}$). Moreover, by construction all remaining leaf vertices are placed in the first quadrant, and by Property~1 of non-strictly slope-disjoint drawings, the slopes of the edges incident to any leaf lie within $\left[0,\frac{\pi}{2}\right]$. Hence, it suffices to show that the rays extending from these leaves (other than at the root) do not intersect.

Let $l_1, l_2, \ldots, l_k$ be the leaves of the tree as visited via in-order traversal that visits the children of each vertex in counter-clockwise order\footnote{We remind readers that throughout this section we use the embedding returned by \cref{algo:convex}, in which the children of every internal vertex $u$ is an ordered set and are placed in \textbf{strictly increasing slope order}: the first child $v_1$ lies in the first quadrant, the second child $v_2$ follows it counter-clockwise, and so on. Consequently, when we perform an in‑order traversal, the leaves are encountered in \textbf{counter-clockwise order around the root}. In particular, $l_1$ is the very first leaf met when we traverse the drawn tree.} and denote by $p_1, p_2, \ldots, p_k$ their respective parents. Note that since the drawing is planar, it suffices to show that 
$slope(p_1,l_1)< slope(p_2,l_2)< \ldots <slope(p_k,l_k)$ for the rays to be non-intersections.\footnote{During an in‐order traversal, the leaves appear in a specific order. When you extend each leaf’s edge as a ray, this ordering guarantees that the rays fan out without overlapping.} By the construction of the angular ranges in \cref{algo:convex} and the non-strict slope-disjoint property, the slopes of the edges $(p_i, l_i)$ satisfy
\[
0 \leq slope(p_1,l_1) < slope(p_2,l_2) < \cdots < slope(p_k,l_k) \leq \frac{\pi}{2}.
\]
Because these slopes are strictly increasing and lie within the interval $\left(0,\frac{\pi}{2}\right)$, the infinite rays obtained by extending each edge $(p_i,l_i)$ are pairwise non-intersecting.
\end{proof}

Observe that, since \cref{algo:convex} reassigns the root of the input tree $T$ to a degree‑one node (a leaf), the resulting drawing is non‑traditional in the sense that  the original root is no longer placed at the plane’s origin $(0,0)$.
If we insist on drawing the root of the original tree on the origin $(0,0)$, we can produce a near‑convex drawing of the tree, i.e., a drawing where each pair of consecutive edges is convex, with the only exception occurring at the root as shown in \cref{thm:monotone_near_convex_root_at_origin}. This is easily done as follows: 
Firstly, we   create a new tree $T'$ by adding to $T$ a new  node $r'$ which is adjacent to root  $r$ and appears first in $r$'s adjacent list.  Secondly, we apply  \cref{algo:convex} on tree $T'$ with $r'$ as its degree-1 root. As a result, we get a convex drawing of $T'$ with the old root $r$ drawn at $(1,1)$ and tree $T$ drawn at the first quadrant having $r$ at its origin.
Finaly, we remove $r'$ and we recenter the drawing so that $r$ is drawn at $(0,0)$. The following theorem is immediate.

\begin{theorem}\label{thm:monotone_near_convex_root_at_origin}
Given an ordered rooted n-vertex tree $T$, we can always produce a  monotone and near-convex planar grid drawing of $T$ that has its root drawn at $(0,0)$, it respects its ordering, and  fits in  a   grid of size  at most $n \times n$.
\end{theorem}

%% file: algo_2quad.tex
\label{sect:2quad_drawing}

In this Section, we examine monotone drawings for unrooted ordered trees. Our approach is to carefully select a vertex $r$ and designate it as the root of the tree. 
The produced tree drawing occupies the first two quadrants, with respect to the location of its root $r$ which is drawn at the origin.   We note that the drawing respects the initial embedding of the tree, that is, the order of the neighbors of each vertex around it is maintained.

The ability to choose a vertex $r$ and to  designate it as the root of the tree, in addition to the use of the first two quadrants, allows us to reduce the used grid to at most $n \times \frac{n}{2}$.

We first describe how to select  the vertex to be designated as the root of the tree. A desirable property of the root node, given the nature of our algorithm, is that its children  have as much balanced subtrees (with respect to  their number of vertices) as possible.

Let $T$ be an unrooted tree. Let $r \in T$ be a vertex such that if we root $T$ at $r$ then for any child  $v$ of $r$, the size of subtree $T_v$ is $|T_v| \leq \frac{n}{2}$. We refer to  $r$  as a \emph{gravity root} of $T$.
Therefore, if an $n$-vertex tree $T$ is rooted at a \emph{gravity root} vertex $r$ there is no vertex $u \in T \backslash r$ such that $|T_u| > \frac{n}{2}$.

\begin{observation}\label{obs:root}
Let $T$ be an $n$-vertex tree, $n >2$, rooted at a gravity root $r$. Then,  $r$  has at least two children.
\end{observation}

\begin{proof}
If we assume that $r$ has only one child,  say $u$, then $|T_u| = |T| - 1 > \frac{|T|}{2}$  (for $n>2$), a contradiction since we assumed $r$ is a gravity root.
\end{proof}

\begin{algorithm}[ht]
\caption{Identify a  gravity root} \label{alg:gr}
\begin{algorithmic}[1]
\Procedure{GravityRootFinder}{T}
\State \hspace*{-0.5cm}Input: A unrooted tree $T$. 
\State \hspace*{-0.5cm}Output: A gravity root vertex $r$.
	\State $r \gets \text{ An arbitary vertex~} u \in T$
	\While {$r$ is not a gravity root}
	\State $u \gets$ the vertex connected to $r$ which lies in the largest connected 
	\State $~~~~~$ component of $T \backslash r$.
	\State $r \gets u$
	\EndWhile
\EndProcedure

\end{algorithmic}
\end{algorithm}

We now show that every tree has a gravity root.

\begin{lemma}
Let $T$ be an $n$-vertex unrooted tree. \Cref{alg:gr} always succeeds in  identifying  a gravity root of $T$ .
\end{lemma}

\begin{proof}
At each iteration, \Cref{alg:gr} gets a step closer to finding a gravity root.
Denote by $T_{lcc}(r)$ the largest connected component of $T \backslash r$. 
By definition, vertex $r$ is a gravity root if $|T_{lcc}(r)| \leq \frac{n}{2}$. We show that at 
each iteration of \Cref{alg:gr} the value of $|T_{lcc}(r)| $ decreases; this continues until a gravity root is reached.

Assume that  $r$ is not a gravity root. Then,  $|T_{lcc}(r)| \geq \frac{n+1}{2}$. Let $u$ be the neighbor of $r$ in  $T_{lcc}(r)$.  
Since \Cref{alg:gr} selects vertex $u$  as the root for next iteration, it is enough to show that 
$|T_{lcc}(u)| < |T_{lcc}(r)|$. 

Note that  the  connected  component of $T \backslash u$ that contains $r$  has size  less or equal than $n-\frac{n+1}{2}=\frac{n-1}{2}$. Thus, if $u$ is not a gravity root then the next candidate gravity root will be a neighbor of $u$ in $T_{lcc}(r)$. Thus, $T_{lcc}(u)$ will be a  proper subtree of $T_{lcc}(r)$, and therefore, $|T_{lcc}(u)| < |T_{lcc}(r)|$.

We conclude  that the value of $|T_{lcc}(r)|$, where $r$ is the candidate gravity root in \Cref{alg:gr} decreases with each iteration until a gravity root is selected.
\end{proof}

By rooting a tree at a \emph{gravity root}, we can obtain a monotone drawing with bounded  angle-range length for any subtree rooted at a child of the root. This is formalized in the  Lemma that follows.
Let function $odd(): \mathbb{N} \rightarrow \left\{ 0,1 \right\}$ evaluate to 1  when its parameter is odd, otherwise it evaluates to  0.

\begin{lemma}\label{lem:angleRangebound}
Let $T$ be an $n$-vertex tree rooted at a \emph{gravity root}  $r$. Let $\left< \theta_1, \theta_2 \right>$ be the angle-range of $r$. \Cref{strategy_main:angle} assigns at each vertex $u\in T \backslash r$ angle-range of length at most $\frac{\theta_2 - \theta_1}{2}  \frac{n - odd(n)}{n - 1}$.
\end{lemma}

\begin{proof}
Let $T$ be an $n$-vertex tree rooted at a \emph{gravity root} $r$.
Since $T$ is rooted at a gravity root then, for any child $u$ of $r$ it holds that $|T_u| \leq \frac{n}{2}$.
Furthermore, if $n$ is odd then it holds that $|T_u| \leq \frac{n-1}{2}$ since the size of a subtree must be an integer.
By making use of the $odd()$ function, we have that for any child $u$ of $r$
it holds that  $|T_u| \leq \frac{n - odd(n)}{2}$.

By \Cref{strategy_main:angle}, we assign to each  child $u$ of $r$ an angle-range of length:
\begin{align*}
a_2(u) - a_1(u) & =  \left( \theta_2 - \theta_1 \right)  \frac{|T_u|}{n-1}\\
				& \leq  \left(\theta_2 - \theta_1 \right)  \frac{\frac{n - odd(n)}{2}}{n-1}\\
				& =   \frac{\theta_2 - \theta_1}{2}  \frac{n - odd(n)}{n-1}
\end{align*}

We complete the proof by noticing that the observation holds not only for the children of $r$ but also for any other vertex of $T \backslash r$. This is due to the fact that \Cref{strategy_main:angle}
always assigns to a vertex of $T$ an angle-range of length equal or smaller to that of its parent.
\end{proof}

In our ``two-quadrant'' algorithm we again use  \Cref{strategy_main:angle} for angle assignment but, we now assign the gravity root of the input tree $T$  angle-range  $\left< 0, \pi \right>$ 
instead of $\left< 0, \frac{\pi}{2} \right>$.
Consequently, in order to  assign grid points to tree vertices we need to extend  \Cref{lemma:distance} 
to cover the case where a vertex has angle-range boundary  $\theta_2 >\frac{\pi}{2}$.

\begin{lemma} \label{lemma:distance2}
Consider angles $\beta_1$, $\beta_2$ with $0 \leq \beta_1 < \beta_2 \leq \pi$.
Then, a grid point $p$  such that  the edge $e$ that connects the origin $(0,0)$ to $p$ satisfies $\beta_1 < slope(e) < \beta_2$, can be identified as follows:

\begin{equation*}
p = \begin{cases}
(0,1)	&	\text{\rm if~}  \beta_1 < \frac{\pi}{2} < \beta_2\\
(x,y)	&	\text{\rm if~} \beta_2 \leq \frac{\pi}{2}, \text{\rm ~where~} (x,y) \text{\rm ~is a valid pair according to \Cref{lemma:distance}}\\
& \text{\rm where~} \theta_1 \gets \beta_1 \text{\rm ~and ~} \theta_2 \gets \beta_2 \\
(-x,y)	&	\text{\rm if~} \beta_1 \geq \frac{\pi}{2}, \text{\rm ~where~}  (x,y) \text{\rm ~is a valid pair  according to \Cref{lemma:distance}}\\
& \text{\rm where~} \theta_1 \gets \pi - \beta_2 \text{\rm ~and~} \theta_2 \gets \pi - \beta_1
\end{cases}
\end{equation*}
\end{lemma}

\begin{proof}
We prove the lemma by taking cases depending on the value of $\beta_1$ and $\beta_2$.
\begin{description}
\item[Case-1: $\beta_1 < \frac{\pi}{2} < \beta_2$.]
It is clear that $\beta_1 < \frac{\pi}{2} < \beta_2 \Leftrightarrow \beta_1 < slope(e) < \beta_2$.

\item[Case-2: $\beta_2 \leq \frac{\pi}{2}$.]
From \Cref{lemma:distance} it holds that: 
$$\theta_1 < slope(e) < \theta_2 \Leftrightarrow \beta_1 < slope(e) < \beta_2.$$

\item[Case-3: $\beta_1 \geq \frac{\pi}{2}$.]
Let $e'$ be the edge that connects the origin to $(x,y)$. Note that  $slope(e) = \pi - slope(e')$.
From \Cref{lemma:distance} it holds that:
\begin{align*}
&\theta_1  <  slope(e') < \theta_2 \\
\Leftrightarrow & \pi - \beta_2  <  slope(e') < \pi - \beta_1 \\
\Leftrightarrow & \beta_1  <  \pi - slope(e') < \beta_2 \\
\Leftrightarrow & \beta_1  <  slope(e) < \beta_2
\end{align*}
\end{description}

\end{proof}

\Cref{algo:2_quadrants} describes our ``two-quadrants'' balanced monotone unrooted-tree drawing algorithm. It consists of three procedures: Procedure \textsc{AssignAngles} (same as in \Cref{algo:1_quadrant}) which assigns angle-ranges to the vertices of the tree according to \Cref{strategy_main:angle}, Procedure \textsc{ExpandedDrawVertices} which assigns each tree vertex to a grid point according to \Cref{lemma:distance2} and Procedure \textsc{UnrootedTreeMonotoneDraw} which assigns a vertex as the root, draws it to point $(0,0)$ with angle-range $\left<0, \pi\right>$ and initiates the drawing of the tree. 

\begin{algorithm}[ht]
\caption{Two-Quadrants Monotone Tree Drawing algorithm}\label{algo:2_quadrants}
\begin{algorithmic}[1]
\Procedure{UnrootedTreeMonotoneDraw}{}
\State  \hspace*{-0.5cm}Input: An $n$-vertex unrooted tree $T$.
\State  \hspace*{-0.5cm}Output: A monotone drawing of $T$ on a grid of size  at most $n \times \frac{1}{2}n$.
\State $r \gets$ \Call {GravityRootFinder}{$T$} (Finds a gravity root as described in
\State  $~~~~~~$ \Cref{alg:gr})
\State $a_1(r) \gets 0, ~a_2(r) \gets \pi$
\State \Call {AssignAngle}{$r, ~a_1(r), ~a_2(r)$}
\State Draw $r$ at $(0,0)$
\State \Call {ExpandedDrawVertices}{$r$}
\EndProcedure

\State
\Procedure{AssignAngles}{u, $a_1$, $a_2$}
\State  \hspace*{-0.5cm}Input: A vertex $u$ and the boundaries of the angle-range $\left<a_1, a_2\right>$ assigned to $u$.
\State  \hspace*{-0.5cm}Action: It assigns angle-ranges to the vertices of $T_u$.
	\For {each child $v_i$ of $u$}
	\State Assign $a_1(v_i),~a_2(v_i)$ as described in \Cref{strategy_main:angle}.
	\State \Call {AssignAngles}{$v_i, ~a_1(v_i), ~a_2(v_i)$}
	\EndFor
\EndProcedure

\State
\Procedure{ExpandedDrawVertices}{u}
\State \hspace*{-0.5cm}Input: A vertex $u$ where $u$ has already been drawn of the grid and angle-ranges have been defined for all vertices of $T_u$.
\State \hspace*{-0.5cm}Action: It draws  the vertices of  $T_u$.

	\For {each child $v_i$ of $u$}
	\State Find a valid pair $(x,y)$  as described in \Cref{lemma:distance2} where
	\State  $~~~~~~\beta_1 \gets a_1(v_i) \text{ and } \beta_2 \gets a_2(v_i)$
	\State If $u$ is drawn at $(u_x, u_y)$, draw $v_i$ at $(u_x+x,u_y+y)$
	\State \Call {ExpandedDrawVertices}{$v_i$}
	\EndFor
\EndProcedure
\end{algorithmic}
\end{algorithm}

The following observation highlights the connection between \Cref{algo:2_quadrants} and \Cref{algo:1_quadrant}.

\begin{observation}\label{obs:alg}
Let  $v$ be a vertex that has been assigned angle-range $\left< a_1(v), a_2(v) \right>$ and let $u$ be its parent which is drawn at grid point $(u_x, u_y)$.   \Cref{algo:2_quadrants} draws $T_v$ in the following way:

\begin{description}
\item [ $a_2(v) \leq \frac{\pi}{2}$:]  \Cref{algo:2_quadrants} draws $T_v$ in the first quadrant, in exactly the same way as \Cref{algo:1_quadrant} does.

\item [$a_1(v) \geq \frac{\pi}{2}$:]   \Cref{algo:2_quadrants} draws $T_v$ in the second quadrant as the reflex drawing (with respect to line $l: x=u_x$) of the drawing \Cref{algo:1_quadrant} produces for $T_u$ if we reverse the order of the children for each vertex $x\in T_v$.

\item [$a_1(v) < \frac{\pi}{2} < a_2(v)$:]  \Cref{algo:2_quadrants} draws $v$ on the $Y$ axis.
Since all children are assigned non-overlapping angle-ranges, at most one child includes $\frac{\pi}{2}$ in its angle-range and, according to the two previous points, the other children are either drawn at the first or second quadrant.
\end{description}
\end{observation}

By combining \Cref{lem:angleRangebound} with \Cref{lemma:distance2}, we  obtain an upper bound on the length of an edge in the drawing produced by \Cref{algo:2_quadrants}.

\begin{lemma} \label{lemma:bound2}
Let $T$ be an $n$-vertex tree rooted at a \emph{gravity root} $r$. Let $v$ be a vertex in $T \backslash r$ with angle-range $\left< \theta_1, \theta_2 \right>$ and let $u$ be its parent. For the vector $e=(x,y)$ that connects  $u$ to $v$, as drawn by \Cref{algo:2_quadrants}, it holds:
\begin{equation*}
max(|x|,y) \leq \frac{\pi}{2} \frac{1}{\theta_2 - \theta_2}  \frac{n - odd(n)}{n-1}
\end{equation*}
\end{lemma}

\begin{proof}
First we note that, according to \Cref{lemma:distance2}, the $y$-coordinate is always positive but the sign of the  $x$-coordinate depends on the angle-range of $v$, as noted in \Cref{obs:alg}.

\begin{description}
\item[$\theta_1 < \frac{\pi}{2} < \theta_2$:] 
The vector that connects $u$ to $v$ is $e=(0,1)$. Therefore,  $max(|x|,y)=1$.
By \Cref{lem:angleRangebound}, and since $v$ is not the tree root,   $v$ has  angle-range length at most $\frac{\pi}{2} \cdot \frac{n-odd(n)}{n-1}$. Therefore:

\begin{align*}
~&\frac{\pi}{2}  \frac{n-odd(n)}{n-1} \geq  \theta_2 - \theta_1 \\ 
\Rightarrow & \frac{\pi}{2}  \frac{n-odd(n)}{n-1}  \frac{1}{\theta_2 - \theta_1}  \geq  1 \\
\Rightarrow & \frac{\pi}{2}  \frac{n-odd(n)}{n-1}  \frac{1}{\theta_2 - \theta_1}  \geq  max(x,y)
\end{align*}

\item[Otherwise:] 
When  $\theta_2 \leq \frac{\pi}{2}$ or $\theta_1 \geq \frac{\pi}{2}$, the grid point assignment is made according to \Cref{lemma:distance2} which, in turn, makes use of \Cref{lemma:distance}. By  applying \Cref{lemma:distance} and by noticing that $\frac{n-odd(n)}{n-1} \geq 1$, the bound is guaranteed.
\end{description}

\end{proof}

\begin{lemma}\label{lemma:algMonotone2}
The drawing produced by \Cref{algo:2_quadrants} is monotone and planar.
\end{lemma}

\begin{proof}
The angle-range assignment of \Cref{strategy_main:angle}  satisfies  Property-2 and Property-3 of the non-strictly slope disjoint drawing as proved in \Cref{lemma:angle}.
In addition, the assignment of the vertices to grid points satisfies  Property-1 of   the non-strictly slope disjoint drawing as proved in \Cref{lemma:distance2}.
Thus, the produced drawing by  \Cref{algo:2_quadrants} is non-strictly slope disjoint and, by \Cref{thm:NSslopeDis_monotonePlanar}, it is monotone and planar.
\end{proof}

It remains to establish a bound on the grid size required by \Cref{algo:2_quadrants}. We consentrate on trees of at least 3 vertices, since it is trivial to draw a tree with two vertices. Our proof uses induction on the number of tree vertices having at least two children.

\begin{lemma} \label{lemma:bound3}
Let $T$ be an $n$-vertex tree, $n >2$,  rooted at a \emph{gravity root} $r$ and $\Gamma$ be the drawing of $T$ produced by \Cref{algo:2_quadrants}. Let $u \in T$ be a vertex  which, in $\Gamma$, is drawn  on the $Y$-axis and consider
 $\phi_u= a_2(u) -a_1(u)$ as assigned by \Cref{algo:2_quadrants}.
Let $\Gamma_u^R$ and $\Gamma_u^L$ be the partial drawings  of $T_u$ that lie in the first and second 
$~~$quadrant\footnote{$~$The positive $Y$-axis in considered to be part of both the first and the second quadrant. So, vertices that are drawn on the $Y$-axis appear in both $\Gamma_u^R$ and $\Gamma_u^L$}, respectively.
Then, each of $\Gamma_u^R$ and $\Gamma_u^L$ uses a grid of side-length  bounded by:
\begin{equation*}
(|T_u|-1)  \frac{\pi}{2}  \frac{n-odd(n)}{n-1}  \frac{1}{\phi_u}
\end{equation*}
\end{lemma}

\begin{proof}
We firstly observe a property that  plays a key role in the proof. All vertices $u \in T$ that are drawn by \Cref{algo:2_quadrants} on the $Y$-axis satisfy, by construction, that $a_1(u) < \frac{\pi}{2} < a_2(u)$. This is due to the fact that a vertex is drawn on the $Y$-axis only if its placement was determined based on the first case of \Cref{lemma:distance2}. 

Secondly, we establish an inequality that holds for any vertex $u \in T\backslash r$.  
Given that \Cref{algo:2_quadrants} assigns 
to the gravity root $r$ angle-range $\left< 0 , \pi \right>$ and since   $u$ is not the gravity root, 
by \Cref{lem:angleRangebound} the angle range $\phi(u)$ of $u$ satisfies 
$\phi(u) \leq \frac{\pi}{2}  \frac{n-odd(n)}{n-1}$. Thus, 
\begin{equation}
1 \leq \frac{\pi}{2}  \frac{n-odd(n)}{n-1}  \frac{1}{\phi(u)}  \label{eq:greatThan1}
\end{equation}

Similar to the proof of \Cref{lemma:bound}, we employ induction on the number of vertices having at least two children. We also make use of the ``edge-length bound'' provided by   \Cref{lemma:bound2}. 
Let $i$ be the number of vertices in $ T_u$ with at least two children.
\begin{description}
\item[Base Case (i=0):] 
In this case,  $T_u$ is just a path and, by \Cref{obs:path}, \Cref{algo:2_quadrants} assigns to every vertex of $T_u$ the same angle-range.
Since  
$a_1(u) < \frac{\pi}{2} < a_2(u)$, for any vertex $v\in T_u$ the vector that connects $v$ to its parent is $e=(0,1)$.
Therefore, by \Cref{algo:2_quadrants},  $T_u$ is drawn on the $Y$-axis and has length $|T_u|-1$. Thus,  both $\Gamma_u^R$ and $\Gamma_u^L$ consist  of only a path of length $|T_u|-1$ which is drawn on the $Y$-axis. 
By \Cref{obs:root}, $u$ is not the gravity root and, thus, we can make use of (\ref{eq:greatThan1}). 
It  immediately follows that  each of $\Gamma_u^R$ and $\Gamma_u^L$ uses a grid of side-length  bounded by:
\begin{equation*}
(|T_u|-1)  \frac{\pi}{2}  \frac{n-odd(n)}{n-1}  \frac{1}{\phi_u}
\end{equation*}

The base case is now settled.

\item[Induction Step:]
We prove the bound only  for the grid side-length of $\Gamma_u^R$ as the case for $\Gamma_u^L$  is symmetric.

We first establish that the only case of interest is when $u$ has two or more  children. 
If $u$ has only one child,  say $v$, then, by \Cref{obs:root} $u$ is not the gravity root. 
By \Cref{obs:path}, $v$ inherits the angle range of its parent and, thus,   $\frac{\pi}{2}$ is contained within $v$'s angle-range. Moreover, $\phi(u)= \phi(v)$. By \Cref{algo:2_quadrants}, the vector that connects $u$ to $v$ is $e=(0,1)$. 
If we assume that the induction hypothesis holds for $v$, then the grid side-length of the $\Gamma_u^R$ is bounded by the grid side-length of $\Gamma_v^R$ plus the length of the vector that connects $u$ to $v$. Therefore, the grid side-length of $\Gamma_u^R$ is bounded by:

\begin{align*}
 & (|T_v|-1)  \frac{\pi}{2}  \frac{n-odd(n)}{n-1}  \frac{1}{\phi_v} + 1 \\
=& (|T_u|-2)  \frac{\pi}{2}  \frac{n-odd(n)}{n-1}  \frac{1}{\phi_u} + 1\\
\overset{\cref{eq:greatThan1}} {\leq}& (|T_u| - 2) \frac{\pi}{2}  \frac{n-odd(n)}{n-1}  \frac{1}{\phi_u} + \frac{\pi}{2}  \frac{n-odd(n)}{n-1}  \frac{1}{\phi_u} \\
=	 & (|T_u|-1)  \frac{\pi}{2}  \frac{n-odd(n)}{n-1}  \frac{1}{\phi_u}
\end{align*}

Therefore, the only case of interest is when $u$ has at least two children.

Let $u \in T$ be a vertex such that $u$ is drawn by \Cref{algo:2_quadrants} on the $Y$-axis,  $u$ has at least two children, and $T_u$ has $i+1$ vertices with at least two children each.
Let $v_1,v_2,\ldots,v_m$ be the children of $u$ such that the drawing of $T_{v_j},~1\leq j\leq m,$  lies on the first quadrant.
By \Cref{obs:alg}, the angle-range of any $v_j$ must be in the form of $\left< a_1(v_j), a_2(v_j) \right>$ where $a_2(v_j) \leq \frac{\pi}{2}$ or $a_1(v_j) < \frac{\pi}{2} < a_2(v_j)$.
We note that the largest grid (wrt its side-length) on the first quadrant devoted to any 
tree\footnote{Recall that  by $T_v^u$ when $v$ is a child of $u$, we denote the tree that consists of edge $(u,v)$ and $T_v$.} 
$T_{v_j}^u, 1\leq j \leq m$, determines the grid side-length of $\Gamma_u^R$  since the subtrees rooted at children of $u$ are drawn completely inside non-overlapping (but possibly touching) angular sectors. 
The above statement holds because all the grids that are used for the subtrees share as common origin vertex $u$ and  we only care about all angular sectors that at least partially lie in the first quadrant.
Therefore, the grid size required to draw $T_u$ is the maximum of the grid sizes required to draw any of $T_{v_j}^u$.

For any vertex $v_j$ with angle-range $\left< a_1(v_j), a_2(v_j) \right>$, if $a_2(v_j) \leq \frac{\pi}{2}$, i.e., $T_{v_j}$ lies entirely in the first quadrant, then, the statement holds from \Cref{lemma:bound} and by noticing that $\frac{n-odd(n)}{n-1} \geq 1$.
For the vertex $v_j$ (there exists at most one such vertex) that $a_1(v_j) < \frac{\pi}{2} < a_2(v_j)$, the number of vertices in $T_{v_j}$ with at least two children is less or equal to $i$, therefore  the induction hypothesis holds for $\Gamma_{v_j}^R$.
Therefore, the statement holds for the first quadrant for any $T_{v_j}$ which is drawn on a grid with grid-length side bounded by,

\begin{equation*}
(|T_{v_j}|-1)  \frac{\pi}{2}  \frac{n-odd(n)}{n-1}  \frac{1}{\phi_{v_j}}
\end{equation*}

For the edge connecting $u$ to $v_j$, by \Cref{lemma:bound2} we require a grid of side-length bounded by,

\begin{equation*}
\frac{\pi}{2}  \frac{n-odd(n)}{n-1} \frac{1}{\phi_{v_j}}
\end{equation*}

Therefore, the total required grid  has side-length bounded by:

\begin{equation*}
|T_{v_j}|  \frac{\pi}{2}  \frac{n-odd(n)}{n-1}  \frac{1}{\phi_{v_j}}
\end{equation*}

Since we employ \Cref{strategy_main:angle}, it holds that:
\begin{equation}
\phi_{v_j}= \frac{|T_{v_j}|}{|T_u|-1}  \phi_u
\label{eq:ac1}
\end{equation}

Thus, the bound on the side-length of the total required grid can be restated as:
\begin{align*}
|T_{v_j}| \frac{\pi}{2} \frac{n-odd(n)}{n-1}  \frac{1}{\phi_{v_j}} 
&\overset{\left(\ref{eq:ac1}\right)}{=} |T_{v_j}|  \frac{\pi}{2}  \frac{n-odd(n)}{n-1}  \frac{1}{ \frac{|T_{v_j}|}{|T_u|-1}  \phi_u }\\
&= (|T_u|-1)  \frac{\pi}{2}  \frac{n-odd(n)}{n-1}  \frac{1}{\phi_{u}}
\end{align*}

Therefore, the statement holds for the induction step. The proof of the lemma is complete.
\end{description}
\end{proof}

We can now state our main result regarding ``two-quadrant" drawings.

\begin{theorem} \label{thm:2Q_algo_gridSize}
Given a rooted $n$-vertex tree $T$, \Cref{algo:2_quadrants} produces a  monotone planar grid drawing using a grid of size  at most:
\begin{align*}
n  \times  \left(\frac{n+1}{2}\right) &\qquad \text{when n is odd}\\
\left( n+1 \right)  \times \left( \frac{n}{2} + 1\right)  &\qquad \text{when n is even}
\end{align*}
\end{theorem}

\begin{proof}
The monotonicity and the planarity of the produced drawing follows directly from \Cref{lemma:algMonotone2}.
By applying \Cref{lemma:bound3} with  gravity root $r$, where \Cref{algo:2_quadrants} assigns $a_1(r) = 0$ and $a_2(r) = \pi$, we get that in the worst case the drawing of $T$ that consists of $\Gamma_r^R$ on the first quadrant and $\Gamma_r^L$ on the second quadrant, uses for each one a grid of side-length that is smaller  or equal to:
\begin{equation*}
\left( n-1 \right) \frac{\pi}{2}  \frac{n- odd(n)}{n-1}  \frac{1}{\pi} = \frac{n - odd(n)}{2}
\end{equation*}

The total width of the grid that \Cref{algo:2_quadrants} draws $T$ is the sum of the width of $\Gamma_r^R$ and $\Gamma_r^L$.
The total height of the grid that \Cref{algo:2_quadrants} draws $T$ is the maximum height of $\Gamma_r^R$ and $\Gamma_r^L$.
Given that a grid of  width $w$ and height $h$ is an $(w+1) \times (h+1)$ grid\footnote{Recall that we measure length (width/height)  in units of distance but, when we denote the dimensions of a grid we use the number of grid points in each dimension.},  the size of the  total grid  used by \Cref{algo:2_quadrants} is bounded by:

\begin{align*}
&\left( 2\left( \frac{n - odd(n)}{2} \right) +1 \right) \times \left( \frac{n - odd(n)}{2} + 1 \right)\\
=&\left( n + 1 - odd(n)\right) \times \left( \frac{n - odd(n)}{2} + 1 \right)
\end{align*}

Therefore, when $n$ is odd the grid size is bounded by $n \times \frac{n+1}{2}$ while, when $n$ is even it is bounded by 
is $(n+1) \times \left( \frac{n}{2}+1\right)$.
\end{proof}

\Crefrange{fig:2Q_algo_binTree}{fig:2Q_algo_worstTree}  present  drawings produced by  \Cref{algo:2_quadrants}. Compare \cref{fig:1Q_algo_binTree,fig:1Q_algo_path}  to \cref{fig:2Q_algo_binTree,fig:2Q_algo_path}, respectively, as they depict drawings of the same trees.  \Cref{fig:2Q_algo_binTree} shows the drawing of a 5-layer complete binary tree (31 vertices). While \Cref{thm:2Q_algo_gridSize} indicates that a grid of size $31 \times 16$ may be required, the binary tree is drawn on a $23 \times 12 $ grid. 
\Cref{fig:2Q_algo_path} shows the drawing of a path (15 vertices). The drawing matches the bound stated in \Cref{thm:2Q_algo_gridSize}.  
Finally, \Cref{fig:2Q_algo_worstTree} shows a drawing of a non-path tree (out of all possible10-vertex rooted trees) that requires maximum area (when produced by \Cref{algo:2_quadrants}).

\begin{figure}[h]
	\centering
	\begin{minipage}[t]{0.30\textwidth}
		\centering
		\includegraphics[scale=0.25]{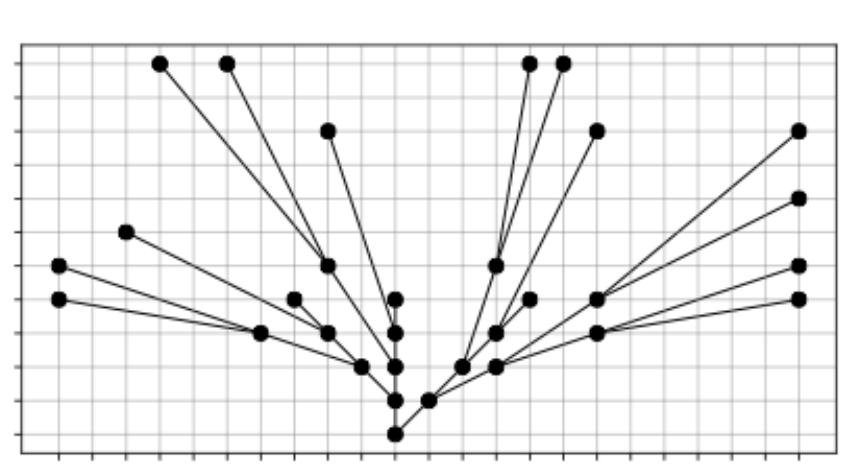}
		\caption{A full binary tree (31 vertices) as drawn by \Cref{algo:2_quadrants}. Grid size: $23 \times 12$.}
		\label{fig:2Q_algo_binTree}
	\end{minipage}
\hfill
	\begin{minipage}[t]{0.30\textwidth}
		\centering
		\includegraphics[scale=0.25]{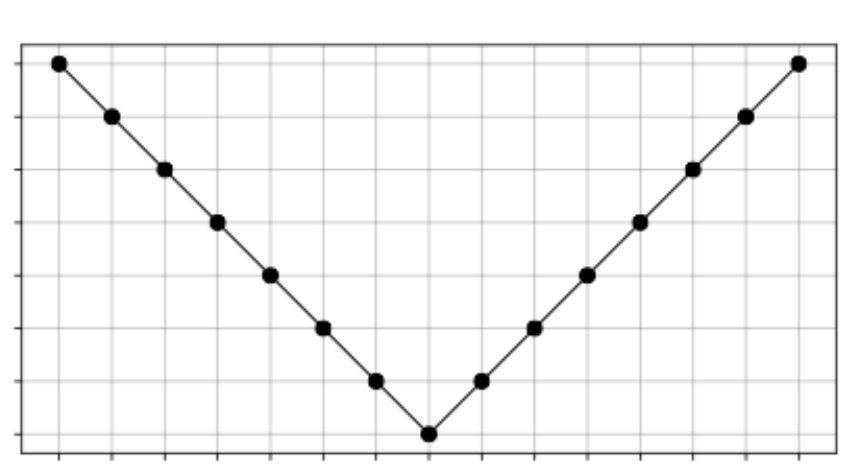}
		\caption{ A path (15 vertices) as drawn by \Cref{algo:2_quadrants}. Grid size: $15 \times 8$.}
		\label{fig:2Q_algo_path}
	\end{minipage}
\hfill
	\begin{minipage}[t]{0.30\textwidth}
		\centering
		\includegraphics[scale=0.20]{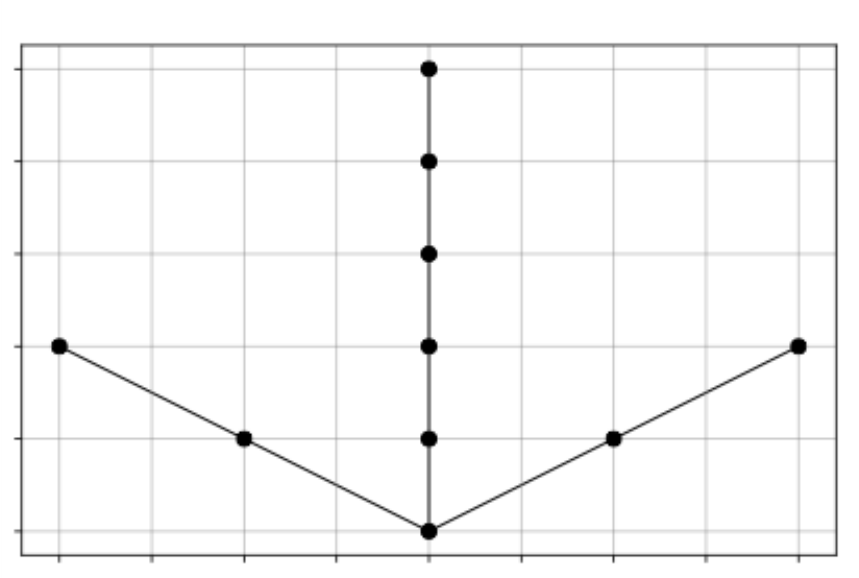}
		\caption{ A non-path tree (10 vertices) with maximum required area when drawn  by \Cref{algo:2_quadrants}. Grid size: $9 \times 6$.}
		\label{fig:2Q_algo_worstTree}
	\end{minipage}
\end{figure}

%% file: algo_4quad.tex
\label{sect:4quad_drawing}

In this Section, we provide an algorithm that construct ``four-quadrants'' drawings of good aspect-ratio for unrooted trees.  \Cref{algo:4_quadrants}, which combines \Cref{algo:1_quadrant} and \Cref{algo:2_quadrants},   yields monotone drawings of $n$-vertex trees on an $\floor{\frac{3}{4} \left(n+2\right)} \times \floor{\frac{3}{4} \left(n+2\right)}$ grid. The main idea of the algorithm is that we first locate a gravity root and partition the subtrees rooted at it into  two groups as balanced as possible and, finally,  draw the subtrees in each group into two disjoint areas. We emphasize that we consider ``non-ordered'' trees,  i.e., our algorithm will not respect (if given) the  embedding of the tree.

\begin{lemma}\label{lemma:group}
Let $T$ be an $n$-vertex tree, $n \geq 3$, rooted at a \emph{gravity root} $r$.
Then, we can identify two subtrees $T_1$ and $T_2$ of $T$ of at most $\frac{2  n + 1}{3}$ vertices each, such that $T_1 \cup T_2 = T$ and   
$T_1 \cap T_2 = r$.
\end{lemma}

\begin{proof}
Since we must have that  $T_1 \cup T_2 = T$ and   $T_1 \cap T_2 = r$, it follows that one of the wanted subtrees, say $T_1$, is formed by $r$ and some of the subtrees rooted at its children, while the other, say $T_2$, is formed by  $r$ and the subtrees rooted at its remaining children.  
Given that  $T$ is rooted at a \emph{gravity root},  the size of each subtree rooted at a child of $r$ is bounded by $\frac{n}{2}$.
Let $m$ be the maximum size of a subtree rooted at a child of $r$, where $m \leq \frac{n}{2}$. We consider cases depending on the value of $m$.

\begin{description}
\item[$\boldsymbol{\frac{n-1}{3} \leq m \leq \frac{n}{2}}$:]
$T_1$ is formed  by $r$ and the subtree of size $m$ that is rooted at a child of $r$. $T_1$ is of size $m+1$. Since $m$ is integer and $m \leq \frac{n}{2}$, it follows that $|T_1| \leq  \frac{2  n + 1}{3}$ for $n \geq 3$.
$T_2$ is formed by $r$ and the subtrees rooted at the remaining children of $r$.  $T_2$ is of size  $n-m$.
Since $ \frac{n-1}{3} \leq m \leq \frac{n}{2} \Leftrightarrow \frac{n}{2} \leq n  - m \leq \frac{2  n + 1}{3}$,  it follows that $|T_2| \leq  \frac{2  n + 1}{3}$. Thus, 
the size of each subtree  is bounded by $\frac{2  n + 1}{3}$ .

\item[$\boldsymbol{m < \frac{n - 1}{3}}$:]
In this case, we form $T_1$ and $T_2$ as follows: Initially, both $T_1$ and $T_2$ consist of the gravity root $r$. We then consider the subtrees rooted at the children of $r$ in increasing order of their size. At any given step, we insert the currently examined subtree   to the smaller of $T_1$ or $T_2$ by attaching it to $r$.  
At the end of this procedure,  the difference in size between the $T_1$ and $T_2$ is at most the size of the biggest subtree rooted at a child of $r$, that is, at most  $m$.
Therefore the size of the largest of $T_1$ and $T_2$  is bounded by $\frac{n-m}{2} + m < \frac{n + m +1}{2} < \frac{2  n + 1}{3}$.
\end{description}
\end{proof}

\Cref{algo:4_quadrants} describes at a high level our \emph{four-quadrant} monotone tree drawing algorithm. 
Let $T$ be the input tree with gravity root $r$. Let $T_1$ and $T_2$, $|T_1|\geq |T_2|$, be the two subtrees of $T$ according to  \Cref{lemma:group}. We draw the tree in two steps. In the first step, we  draw $T_1$ according to the ``two-quadrands'' \Cref{algo:2_quadrants}. In doing so, we take special care to   place the path from the gravity root $r^\prime$ of $T_1$  to $r$ on the $X$-axis with the appropriate change in the embedding of $T_1$.
In the second step, we draw $T_2$ according to the traditional ``one-quadrant'' \Cref{algo:1_quadrant}.
Then,  we combine the drawing of $T_1$  with the reflect on the x axis of the drawing of $T_2$.
The way we combine the two drawings  is demonstrated in \Cref{fig:4_quadrands}.
The drawing produced is monotone and its grid size is bounded by
 $\floor{\frac{3}{4} \left(n+2\right)} \times \floor{\frac{3}{4} \left( n+2\right)}$.

\begin{figure}[t]
	\centering
	\begin{minipage}{0.7\textwidth}
		\centering
		\includegraphics[scale=0.45]{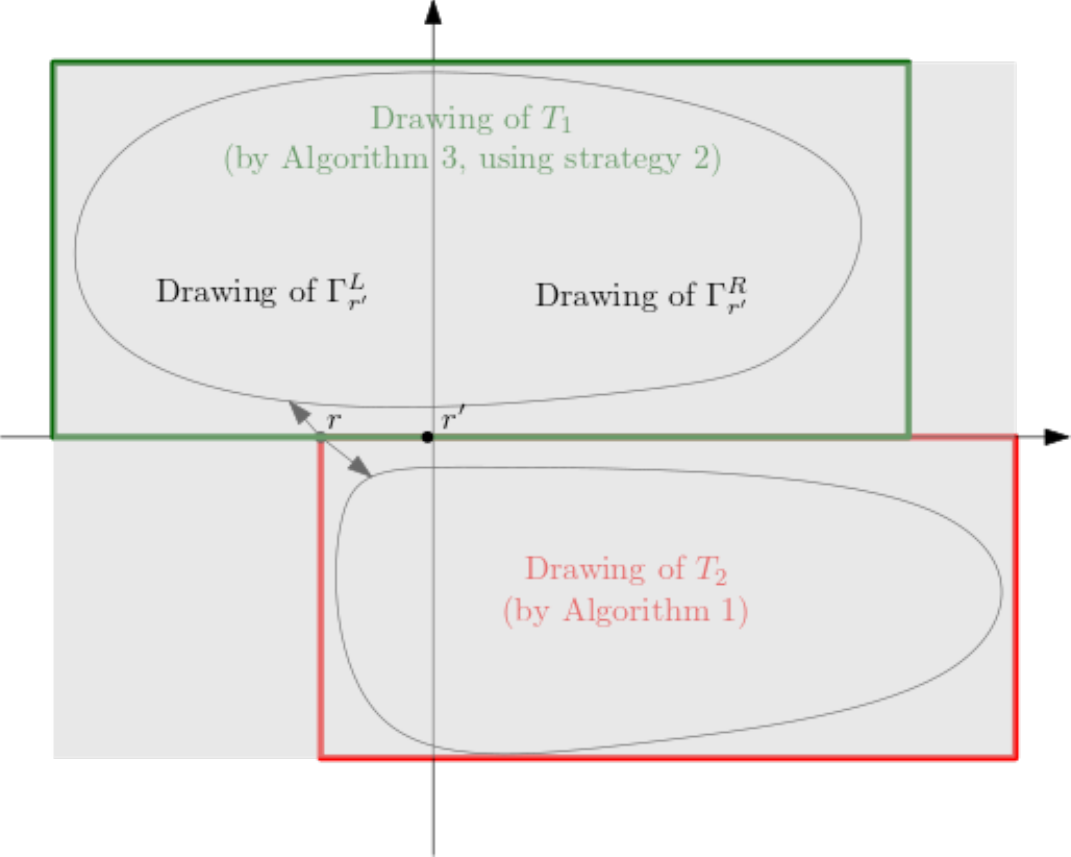}
		\caption{Example of how does \Cref{algo:4_quadrants} places $T_1$ and $T_2$.}
		\label{fig:4_quadrands}
	\end{minipage}
\end{figure}

The next strategy explains how to change the embedding of $T_1$ in order to place  the path from the gravity root $r^\prime$ of $T_1$ to $r$ on the x-axis and to the left of $r^\prime$.

\begin{strategy} \label{strategy:mod}
In  $T_1$,   place each vertex in the path from the gravity root $r^\prime$ of $T_1$ to $r$ as the last child of its parent. In that way, \Cref{strategy_main:angle} (which is employed by \Cref{algo:2_quadrants}) assigns each vertex from $r^\prime$ to  $r$ angle-ranges in the from of $\left< \theta_1, \pi \right>$.
Moreover,  assign slope $\pi$ to each edge in the path. In that way, the whole path from  the gravity root $r^\prime$ of $T_1$ to $r$  is drawn on the $X$-axis.  
\end{strategy}

Observe that, when we apply \Cref{strategy:mod} and draw  tree $T_1$ based on \Cref{algo:2_quadrants}, \Cref{thm:2Q_algo_gridSize} which bounds the drawing area still holds.
This is due to the facts that  (i) each edge $e$  that connects  a vertex $u$ (on the path from $r^\prime$ to $r$)  to its child  with slope $slope(e) = \pi$ is drawn at the boundary of the the angle-range of  $u$,   and (ii) the length of such an edge $e$ is 1, i.e.,  the least length possible. Of course, the drawing remains monotone  and planar as the following lemma indicates.

\begin{algorithm}[ht]
\caption{Four-Quadrants Monotone Tree Drawing algorithm}\label{algo:4_quadrants}
\begin{algorithmic}[1]
\Procedure{4QuadrantTreeMonotoneDraw}{}
\State  \hspace*{-0.5cm}Input: An $n$-vertex unrooted tree $T$.
\State  \hspace*{-0.5cm}Output: A four-quadrant monotone drawing of $T$ on a grid of size  at most $\floor{\frac{3}{4} \left(n+2\right)} \times \floor{\frac{3}{4} \left(n+2\right)}$.
\State 
\State Find $T_1$ and $T_2$ according to \Cref{lemma:group}, where $|T_1| \geq |T_2|$.
\State Draw $T_1$ according to \Cref{algo:2_quadrants} with the modification of \Cref{strategy:mod}.
\State Draw $T_2$ according to \Cref{algo:1_quadrant}.
\State Combine the drawing of $T_1$ with the reflect on the $X$-axis  drawing of $T_2$.
\EndProcedure
\end{algorithmic}
\end{algorithm}

\begin{lemma}
\label{lem:alg_4Q_monotone}
The drawing produced by \Cref{algo:4_quadrants} is planar and monotone.
\end{lemma}

\begin{proof}
We prove the lemma by showing that the unique simple path that connects two arbitrary vertices $u,~v$ of tree $T$ is monotone with respect to some direction. This will imply the monotonicity of the drawing of $T$ and, by \Cref{t:b}, its planarity.

Consider the drawing of an arbitrary tree $T$ produced by \Cref{algo:4_quadrants}, and let $u,~ v$ be two arbitrary vertices of $T$.
\begin{description}
\item[Case 1: $u \in T_2$ and $v \in T_2$.]
By  \Cref{lemma:algMonotone}  the drawing of $T_2$ is monotone. Given that the simple path from $u$ to $v$ is entirely contained in $T_2$, the path is monotone.

\item[Case 2: $u \in T_1$ and $v \in T_1$.]
If $T_1$ was drawn by \Cref{algo:2_quadrants} (as it is described in \Cref{sect:2quad_drawing}) then, by \Cref{lemma:algMonotone2}, the drawing of $T_1$ would be monotone. Thus, the simple path from $u$ to $v$ would also be monotone since it is entirely contained in $T_1$. However,  \Cref{algo:4_quadrants}   additionally applies \Cref{strategy:mod} when drawing $T_1$, and thus, we have to ensure that the changes in the drawing due to \Cref{strategy:mod} do not affect its monotonicity.

Let $r$ and $r^\prime$ be the gravity roots of $T$ and $T_1$, respectively. The only edges that  violate the non-strictly slope-disjoint property of the produced drawing are those that enter nodes on the path from $r^\prime$ to $r$. Recall that, by \Cref{strategy:mod} all  these edges lie on the $X$-axis and have slope $\pi$. Let $e=(u,v)$ be such an  edge in $T_1$ from vertex $u$ to its child $v$  of slope $\pi$.

This is a violation to Property-1 of non-strictly slope disjoint drawings (see \Cref{def:nonStrictlySlopeDisjoint}). Property-1 requires that  every edge $e$ from  a vertex $u$ to any of its children has a slope that falls within the angle-range of $u$ and does not take the boundary values, that is, $a_1(u) < slope(e) < a_2(u)$.  In our case, by \Cref{strategy:mod} we have that $a_2(u) = \pi$ and $slope(e) = a_2(u) = \pi$.
Therefore, for edge $e=(u,v)$ it holds that $a_1(u) < slope(e) \leq a_2(u)$.

We can rotate the whole drawing clockwise around the gravity root $r^\prime$ of $T_1$ by an arbitrarily small amount $\epsilon>0$.
If we denote by $slope(e)$ the slope of the edge $e$ in the original drawing and by $slope'(e)$ the slope of the edge $e$ in the new rotated drawing, it holds that $slope'(e) = slope(e) - \epsilon$.

For any vertex $u$ with angle-range in the form of $\left<a_1(u), a_2(u) \right>$, for any  edge $e= (u,v)$ that connects $u$ to its child $v$, since $\epsilon >0$ is arbitrarily small,  it holds that:
\begin{align*}
& a_1(u)  < slope(e)  \leq a_2(u) \\
\Rightarrow &a_1(u) -\epsilon < slope'(e) \leq a_2(u) - \epsilon \\
\Rightarrow &a_1(u)  < slope'(e)  < a_2(u)
\end{align*}

So, for the slightly rotated drawing, all the properties of non-strictly slope disjoint drawings are satisfied and,  by \Cref{thm:NSslopeDis_monotonePlanar},  the drawing of $T_1$ is monotone and planar.

\item[if $u \in T_1$ and $v \in T_2$:] The simple path from vertex $u$ to vertex $v$  
is the concatenation of the simple path from $u$ to the gravity root $r$  of $T$ and of the simple path from $r$ to $v$.
If we consider  $r$ as the origin, the edges from  $u$ and $r$ lie inside the first two quadrants, with the exception of the edges between $r$ and the gravity root $r^\prime$ of $T_1$ which lie on the 
$X$-axis, while the edges between $r$ and $v$ lie  inside the fourth quadrant.
 
It is easy to observe that the combined path is monotone with respect to a line with slope $\frac{\pi}{2} + \epsilon$ where $\epsilon>0$ is arbitrarily small. Crucial to this observation is that the two drawing do not overlap. Indeeed, the drawing of $T_2$, as it is drawn with \Cref{algo:1_quadrant}, it lies entirely in the fourth quadrant (with respect to $r$) and non of its vertices lies on the $X$-axis.
\end{description}
From the three cases above, we conclude that the produced drawing is monotone and planar. This completes the proof.
\end{proof}

\begin{theorem} \label{thm:4Q_algo_gridSize}
Given an  $n$-vertex Tree $T$, \Cref{algo:4_quadrants} draws $T$ in a grid of size  at most $\floor{\frac{3}{4} \left(n+2\right)} \times \floor{\frac{3}{4} \left( n+2\right)}$.
\end{theorem}
\begin{proof}
Let $r$ and $r^\prime$ be the vertices used by \Cref{algo:4_quadrants} as the gravity roots of $T$ and  $T_1$, respectively.  Based on  the modification of the drawing of $T_1$ by \Cref{strategy:mod}, $r$ lies in the second quadrant if we assume $r^\prime$ as the origin node.
Furthermore, $T_2$ is drawn in the fourth quadrant if we assume $r$ as the origin node. From \Cref{fig:4_quadrands}, it is clear that the worst case grid size for the combined drawing is realized
 when  $r^\prime$ coincides with $r$.

By \Cref{thm:2Q_algo_gridSize}, the grid side-length of subdrawings $\Gamma_{r'}^R$ in the first quadrant and $\Gamma_{r'}^L$ in the second quadrant is bounded by $\frac{|T_1|}{2}$ while in  the fourth quadrant,  according to \Cref{thm:1Q_algo_gridSize}, the side-length   is  at most $|T_2|-1$.

Therefore, the grid width is at most $max(|T_1|,  \frac{|T_1|}{2}  + |T_2|-1)$ and the grid height is at most $ \frac{|T_1|}{2} + |T_2|-1$.
We consider two  cases depending on  whether $\frac{|T_1|}{2} > |T_2|-1$.\\

\begin{description}
\item[Case 1: $\boldsymbol{\frac{|T_1|}{2} > |T_2|-1}$.~]
It this case, it is clear that  both the width and the height of the  drawing are bounded by $|T_1|$.
Given that the gravity root $r$ is included in both $T_1$ and $T_2$, we have that:
\begin{equation} \label{eq:f1}
|T_1|+|T_2| = n+1 
\end{equation}
From the assumption, we have:
\begin{align*}
& \frac{|T_1|}{2} > |T_2|-1 \\
~~\Rightarrow~~ & \frac{|T_1| + |T_2|}{2} > \frac{3}{2}  |T_2| - 1 \\
\overset{\left(\ref{eq:f1} \right) }{~~\Rightarrow~~} & \frac{n+1}{2} > \frac{3}{2}  |T_2| - 1 \\
~~\Rightarrow~~ & \frac{1}{3}  n + 1 > |T_2|
\end{align*}

Furthermore, we also have that:
\begin{align*}
|T_2| &=  n+1 - |T_1| \\
&\hspace*{-.65cm}\overset{\left( \Cref{lemma:group} \right)}  {\geq}  n + 1 - \left( \frac{2}{3}  n +\frac{1}{3} \right)\\
& \geq  \frac{1}{3}  n + \frac{2}{3} 
\end{align*}

Therefore,  since $\frac{n}{3} +\frac{2}{3} \leq |T_2| < \frac{n}{3} + 1$, the only integer that satisfies this set of inequalities is $|T_2| = \frac{1}{3}  n + \frac{2}{3}$.
So, by~(\ref{eq:f1}), $|T_1| = n+1 - \left( \frac{1}{3} n +\frac{2}{3} \right) = \frac{2}{3} n + \frac{1}{3}$.
Thus, the required grid is of size\footnote{Recall that, in general, a drawing of length $l$ and width $w$ is drawn in a grid of dimensions $(l+1) \times (w+1)$. } at most:
\begin{equation*}
\left( \frac{2}{3} n + \frac{4}{3}\right) \times \left( \frac{2}{3} n + \frac{4}{3}\right)
\end{equation*}

This grid,  for any $n\geq 1$, fits in  a grid of dimensions: 
\begin{equation*}
\left( \frac{3}{4} (n + 2)\right) \times \left( \frac{3}{4} (n + 2) \right)
\end{equation*}

Therefore, the statement holds.\\

\item[Case 2: $\boldsymbol{ \frac{|T_1|}{2} \leq |T_2|-1}$.~]
In this case, the grid side-length is: 
\begin{align*}
\frac{|T_1|}{2}+|T_2|-1 &=\frac{|T_1|+|T_2|}{2} + \frac{|T_2|}{2} - 1\\
&\hspace*{-.17cm}\overset{\left( \ref{eq:f1} \right)~} {=} \frac{n+1}{2} + \frac{|T_2|}{2} - 1 \\
&\hspace*{-.45cm}\overset{|T_2| \leq |T_1|}{\leq}\enspace  \frac{n+1}{2} + \frac{n+1}{4} - 1 \\
&= \frac{3n-1}{4}
\end{align*}

Therefore, in this case the required grid is of size
\begin{equation*}
\left( \frac{3}{4}  n + \frac{3}{4} \right) \times \left(\frac{3}{4}  n + \frac{3}{4}\right)
\end{equation*}

which, obviously, fits in a grid of size
\begin{equation*}
\left( \frac{3}{4}  (n + 2) \right) \times \left(\frac{3}{4}  (n + 2)\right)
\end{equation*}

\end{description}

Since the bound for the grid size must be integer, the floor of the bound also bounds the grid size. Therefore, as stated in the lemma, the required grid is of size:
\begin{equation*}
\left\lfloor \frac{3}{4}  \left( n+2 \right) \right\rfloor \times \left\lfloor \frac{3}{4}  \left( n + 2 \right) \right\rfloor
\end{equation*}
~\end{proof}

\crefrange{fig:4Q_algo_binTree}{fig:4Q_algo_worstTree}  present  drawings produced by  \Cref{algo:4_quadrants}. In the drawings, we indicate by  a solid square (rhombus)  the gravity root of tree $T$ (resp., $T_1$).  \Cref{fig:4Q_algo_binTree} shows a drawing of aspect-ratio equal to one for  a 5-layer complete binary tree (31 vertices). 
While \Cref{thm:4Q_algo_gridSize} indicates that a grid of size $24 \times 24$ may be required, the binary tree is drawn on a $17 \times 17 $ grid. 
\Cref{fig:4Q_algo_path} shows the drawing of a path (15 vertices). 
Finally, \Cref{fig:2Q_algo_worstTree} shows a drawing of a non-path tree (out of all possible 10-vertex rooted trees) that requires maximum area (when produced by \Cref{algo:2_quadrants}). We have drawn all 10-vertex rooted trees and have identified non-path trees that require  the maximum area. While  \Cref{thm:4Q_algo_gridSize} indicates that an $9 \times 9$ grid may be used for a tree of 10 vertices, the drawing of maximum area uses a grid of size $8 \times 7$.

\begin{figure}[h]
	\centering
	\begin{minipage}[t]{0.48\textwidth}
		\centering
		\includegraphics[scale=0.35]{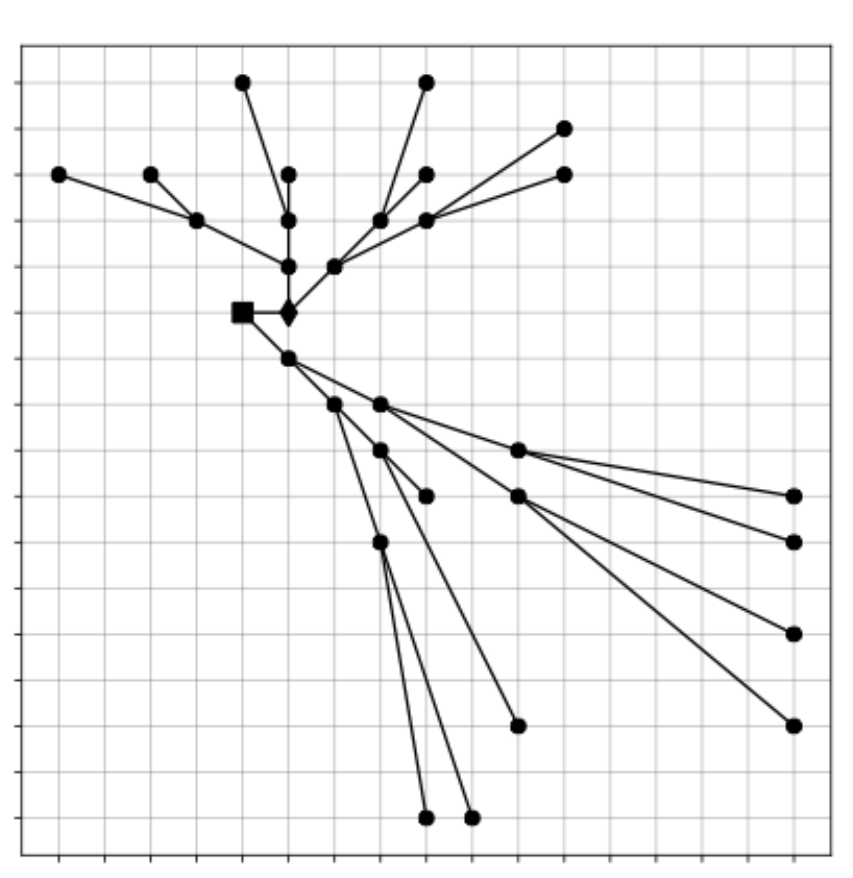}
		\caption{A full binary tree (31 vertices) as drawn by \Cref{algo:4_quadrants}. Grid size: $17 \times 17$.}
		\label{fig:4Q_algo_binTree}
	\end{minipage}
\hfill
	\begin{minipage}[t]{0.45\textwidth}
		\centering
		\includegraphics[scale=0.35]{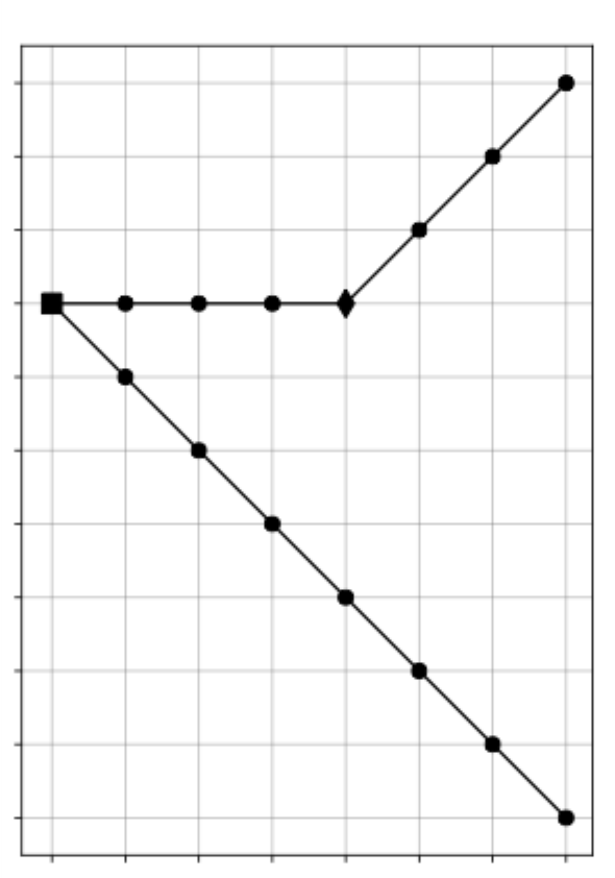}
		\caption{ A path (15 vertices) as drawn by \Cref{algo:4_quadrants}. Grid size: $8 \times 11$.}
		\label{fig:4Q_algo_path}
	\end{minipage}
\end{figure}

\begin{figure}[h]
		\centering
		\includegraphics[scale=0.35]{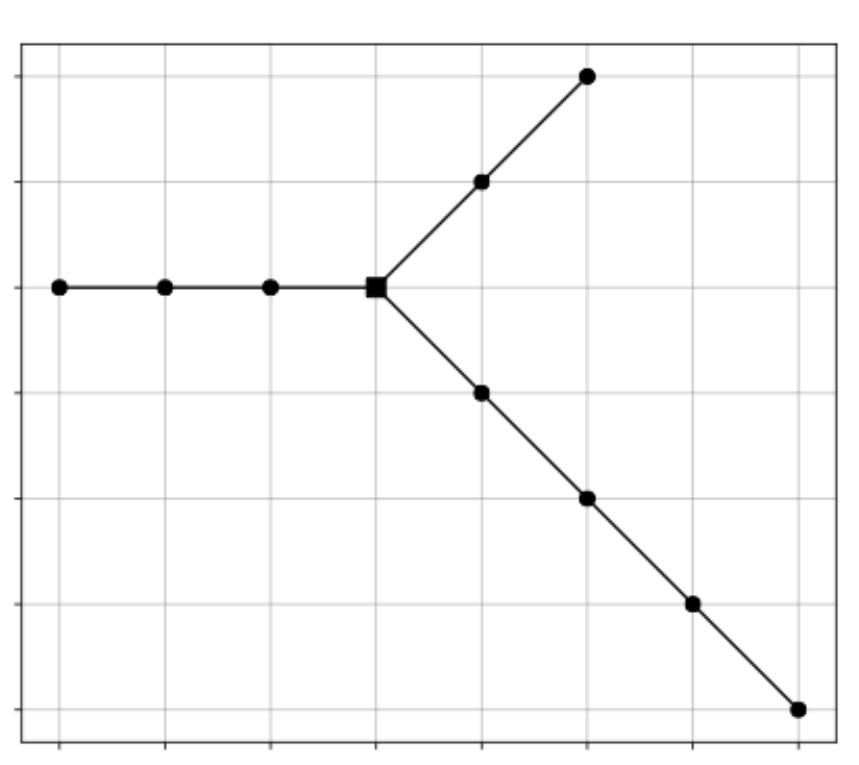}
		\caption{ A non-path tree (10 vertices) with maximum required area when drawn  by \Cref{algo:4_quadrants}. Grid size: $8 \times 7$.}
		\label{fig:4Q_algo_worstTree}
\end{figure}

%% file: conclusion.tex
\label{sect:conclusions}

We have described several algorithms that produce monotone drawings of trees. The algorithm that has the best aspect ratio produces a monotone drawing of  an $n$-vertex tree on a grid of size at most  
$\floor{\frac{3}{4} (n+2)} \times \floor{\frac{3}{4} (n+2)}$.  The following problems on monotone tree drawings are worth studying:
\begin{enumerate}

\item He and He~\cite{mt:4} described a tree that requires for its monotone drawing a grid of size at least $\frac{n}{9} \times \frac{n}{9}$. Can this bound be improved? Is there a tree that requires a larger grid for its monotone drawing?

\item The angular resolution of the produced drawing has not been studied. Is there a trade-off between the angular resolution and the grid size of the monotone drawing? 

\end{enumerate}